\documentclass{siamart171218}

\usepackage[xindy]{glossaries}
\usepackage{xcolor}
\usepackage{amsfonts,mathtools,bm}
\usepackage[OT1]{fontenc}
\usepackage[inline]{enumitem}
\usepackage[caption=false]{subfig}

\glsdisablehyper

\newsiamremark{example}{Example}
\newsiamremark{remark}{Remark}

\def\eqns#1{\begin{equation*}#1\end{equation*}}
\def\eqnl#1#2{\begin{equation}\label{#1}#2\end{equation}}

\def\eqnsa#1{\begin{subequations}\begin{align*}#1\end{align*}\end{subequations}}
\def\eqnmla#1#2{\begin{subequations}\label{#1}\begin{align}#2\end{align}\end{subequations}}

\def\d{\mathrm{d}}

\def\s{\mathrm{s}}

\def\bsG{\bm{G}}
\def\bsI{\bm{I}}
\def\bsp{\bm{p}}

\def\bsx{\bm{x}}
\def\bsX{\bm{X}}
\def\bsz{\bm{z}}
\def\bsZ{\bm{Z}}
\def\bseta{\bm{\eta}}
\def\bspi{\bm{\pi}}

\def\calB{\mathcal{B}}
\def\calC{\mathcal{C}}

\def\calN{\mathcal{N}}

\def\calO{\mathcal{O}}
\def\calS{\mathcal{S}}
\def\calT{\mathcal{T}}
\def\calU{\mathcal{U}}
\def\calV{\mathcal{V}}

\def\bbE{\mathbb{E}}
\def\bbN{\mathbb{N}}
\def\bbM{\mathbb{M}}
\def\bbP{\mathbb{P}}
\def\bbR{\mathbb{R}}

\DeclareMathOperator*{\argmin}{argmin}
\DeclareMathOperator{\Var}{Var}
\DeclareMathOperator{\erf}{erf}
\DeclareMathOperator{\Lip}{Lip}

\def\AND{\qquad\mbox{ and }\qquad}

\def\defeq{\doteq}
\def\given{\,|\,}  % Conditioning
\def\pf{\#}  % Push forward
\def\tr{\mathrm{t}}  % Transpose

\newacronym{sde}{SDE}{stochastic differential equation}
\newacronym{mse}{MSE}{mean squared error}
\newacronym{mlmc}{MLMC}{multilevel Monte Carlo}
\newacronym{mlpf}{MLPF}{multilevel particle filter}
\newacronym{pf}{PF}{particle filter}
\newacronym{ml}{ML}{multilevel}
\newacronym{wrt}{w.r.t.\@}{with respect to}

\title{Multilevel Monte Carlo for Smoothing via Transport Methods}

\author{Jeremie Houssineau%
\thanks{DSAP, National University of Singapore. Email: \href{mailto:stahje@nus.edu.sg}{stahje@nus.edu.sg}}
\and
Ajay Jasra%
\thanks{DSAP, National University of Singapore. Email: \href{mailto:staja@nus.edu.sg}{staja@nus.edu.sg}}
\and
Sumeetpal S.\ Singh%
\thanks{Department of Engineering, University of Cambridge  and The Alan Turing Institute. Email:~\href{mailto:sss40@cam.ac.uk}{sss40@cam.ac.uk}}
}

\headers{MLMC for Smoothing via Transport}{J.\ Houssineau, A.\ Jasra and S.S.\ Singh}

\begin{document}

\maketitle

\begin{abstract}
In this article we consider recursive approximations of the smoothing distribution associated to partially observed \glspl{sde}, which are observed discretely in time.  Such models appear in a wide variety of applications including econometrics, finance and engineering. This problem is notoriously challenging, as the smoother is not available analytically and hence require numerical approximation. This usually consists by applying a time-discretization to the \gls{sde}, for instance the Euler method, and then applying a numerical (e.g.\ Monte Carlo) method to approximate the smoother. This has lead to a vast literature on methodology for solving such problems, perhaps the most popular of which is based upon the \gls{pf} e.g.\ \cite{Doucet2011}. In the context of filtering for this class of problems, it is well-known that the particle filter can be improved upon in terms of cost to achieve a given \gls{mse} for estimates. This in the sense that the computational effort can be reduced to achieve this target \gls{mse}, by using \gls{ml} methods \cite{Giles2008,Giles2015,Heinrich2001}, via the \gls{mlpf} \cite{Gregory2016,Jasra2015,Jasra2018}. For instance, to obtain a \gls{mse} of $\mathcal{O}(\epsilon^2)$ for some $\epsilon>0$ when approximating filtering distributions associated with Euler-discretized diffusions with constant diffusion coefficients, the cost of the \gls{pf} is $\mathcal{O}(\epsilon^{-3})$ while the cost of the \gls{mlpf} is $\mathcal{O}(\epsilon^{-2}\log(\epsilon)^2)$. In this article we consider a new approach to replace the particle filter, using transport methods in \cite{Spantini2017}. In the context of filtering, one expects that the proposed method improves upon the \gls{mlpf} by yielding, under assumptions, a \gls{mse} of $\mathcal{O}(\epsilon^2)$ for a cost of $\mathcal{O}(\epsilon^{-2})$.
This is established theoretically in an ``ideal'' example and numerically in numerous examples.
\end{abstract}



%%%%%%%%
\section{Introduction}

The smoothing problem often refers to the scenario where one has an unobserved Markov chain (or signal) in discrete or continuous time and one is interested in inferring the hidden process on the basis of observations, which depend upon the hidden chain. The case we consider is where the hidden process follows a \gls{sde} and the observations are regularly recorded at discrete times; given the signal at a time $t$ the observation is assumed to be conditionally independent of all other random variables. The process of filtering is to infer some functional of the hidden state at time $t$ given all the observations at time $t$ and the smoothing problem to infer some functional of potentially all the states at the  discrete observation times again given all the observations. It is often of interest to do this recursively in time. This modelling context is relevant for many real applications in econometrics, finance and engineering; see e.g.\ \cite{Cappe2005} and the references therein.

The smoothing problem is notoriously challenging. Supposing one has access to the exact transition of the \gls{sde}, then unless the observation density is Gaussian and depends linearly on the hidden state and the transition density is also Gaussian depending linearly on the previous state, the filter and smoother are not analytically tractable (unless the state-space of the position of the diffusion at any given time is finite and of small cardinality); see \cite{Crisan2008}. However, it is seldom the case that even the transition density (or some unbiased approximation of it, e.g.\ \cite{Fearnhead2008} and the references therein) is available; this is assumed throughout the article. Thus typically, one time-discretizes the diffusion process and then one seeks to perform filtering and smoothing from the time-discretized model. This latter task is still challenging as it is still analytically intractable. There is a vast literature on how to numerically approximate the filter/smoother (e.g.\ \cite{Crisan2011}) and perhaps the most popular of which is the particle filter. This is a method whose cost grows linearly with the time parameter and generates $N$ samples in parallel. These samples are put through sampling and resampling operations. It is well-known that when estimating the filter, the error is uniform in time. For the smoother, the error often grows due to the so-called path degeneracy problem and indeed, there are many smoothing problems for which it is not appropriate; see \cite{Kantas2015} for some review and discussion. In the context of the problem in this article, when only considering the filter, ignoring the time parameter and under assumptions, to  obtain a \gls{mse} of $\mathcal{O}(\epsilon^2)$ for some $\epsilon>0$ the cost of the \gls{pf} is $\mathcal{O}(\epsilon^{-3})$. The \gls{mse} takes into account the exact filter (i.e.~the one with no time discretization).

\Gls{mlmc} methods \cite{Giles2008, Giles2015, Giles2009, Giles2014, Heinrich2001} are of interest in continuum systems which have to be discretized in one dimension, just as in this article (extensions to discretization in multiple dimensions have been proposed and studied in \cite{Crisan2017,Haji2016}). We explain the idea informally as follows: let the time parameter be fixed and denote by $p^L_t$ the filter associated to a (say Euler) discretization level $h_L>0$, set $X_t\in\mathbb{R}^d$, $d\geq 1$  and for $\varphi:\mathbb{R}^d\rightarrow\mathbb{R}$ bounded denote by $p^L_t(\varphi)$ the expectation of $\varphi$ \gls{wrt} the filter. Then the \gls{mlmc} method is based upon the following approach. Consider $0<h_L<h_{L-1}<\cdots<h_0<+\infty$ a sequence of discretizations, where $h_L$ is the most accurate (finest) discretization and $h_0$ the least (coarsest), the \gls{ml} identity is
\eqns{
p^L_t(\varphi) = \sum_{l=0}^L ( p^l_t - p^{l-1}_t)(\varphi)
}
where $p^{-1}_t$ is an arbitrary measure satisfying $p^{-1}_t(\varphi) = 0$ for every $\varphi$. The idea is then to sample $N_0$ independent samples from $p^0_t$ and then, independently for each $1\leq l \leq L$ independently sample $N_l$ coupled pairs from the pair $(p^l_t,p^{l-1}_t)$. The \gls{mlmc} estimator is then
\eqns{
\frac{1}{N_0}\sum_{i=1}^{N_0} \varphi(X^0_{t,i}) + \sum_{l=1}^L\frac{1}{N_l}\sum_{i=1}^{N_l}[\varphi(X^l_{t,i})-\varphi(X^{l-}_{t,i})]
}
where $\{X^0_{t,i}\}_{i=1}^{N_0}$ are i.i.d.\ $p^0_t$ and $\{(X^l_{t,i},X^{l-}_{t,i})\}_{i=1}^{N_l}$ are i.i.d.\ from a coupling of $(p^l_t,p^{l-1}_t)$. To obtain a \gls{mse} of $\mathcal{O}(\epsilon^2)$ one sets $L$ such that the squared bias is $\mathcal{O}(\epsilon^2)$ (the bias is known in the context of interest). If one has $\Var(\varphi(X^l_{t,1})-\varphi(X^{l-}_{t,1})) = \mathcal{O}(h_l^{\beta})$ for some $\beta>0$ then one can try to minimize (\gls{wrt} $N_1,\dots,N_L$) the cost $\sum_{l=1}^LN_lh_l^{-\zeta}$ ($\zeta=1$ for an Euler discretization) subject to the variance $1/N_0 + \sum_{l=1}^L h_l^{\beta}/N_l$ being $\mathcal{O}(\epsilon^2)$. \cite{Giles2008} finds a solution to this problem. The main issue in the context of smoothing, is that one (typically) does not know how to sample from the smoothers nor the couplings.

In \cite{Gregory2016,Jasra2015,Jasra2018} it is shown how to utilize the \gls{pf} to leverage on the potential decrease in cost to obtain a given \gls{mse}. This has been termed the MLPF. The idea is to use couplings in the Euler dynamics and the resampling operation of a \gls{pf}. This has been later refined in \cite{Sen2018}. To our knowledge, the only theoretical work for the \gls{mlpf} in \cite{Jasra2015}, shows that to obtain a \gls{mse} of $\mathcal{O}(\epsilon^2)$ the cost in \gls{mlpf} is $\mathcal{O}(\epsilon^{-2}\log(\epsilon)^2)$, for some specific (constant diffusion coefficient) models and under particular assumptions.  This is known to be worse than the rates obtained in \cite{Giles2008} in the case where there are no observations. Here and throughout, the time parameter is omitted from the discussion on cost and error, despite the fact that these are important considerations in general.

The main idea in this article is to adopt an alternative method to the \gls{pf}. The approach is to use transport methods \cite{Spantini2017}. Transport maps have been used for Bayesian inference \cite{ElMoselhy2012, Heng2015} and more specifically for parameter estimation in \cite{Parno2016} based on a related multi-scale idea. The basic idea is to obtain a map such that the image of samples from an easy-to-sample distribution through this map has exactly the type which one desires. In \cite{Spantini2017} it is shown how to develop numerical approximations of maps, associated exactly to the distributions of interest in this article. These approximations often induce i.i.d.~Monte Carlo approximations of expectations of interest, albeit with a numerical error associated to the approximation of the transport map.  As mentioned in \cite{JasLaw2017}, it is simple to induce coupled pairs using the method of \cite{Spantini2017} and this is exactly what is done in this paper. The potential advantages of this method relative to the \gls{mlpf} are then as follows:
\begin{enumerate}[label=(\roman*)]
\item \label{it:coupledResampling} The \gls{ml} rate lost by coupled resampling can be regained in the context of filtering.
\item \label{it:smoother} The method can be used for approximating the expectation of some functionals \gls{wrt} the smoother, whereas the approach in \cite{Jasra2015,Jasra2018} is typically not useful for smoothing at large time-lags.
\end{enumerate}
In this article we establish that \ref{it:coupledResampling} can hold in an ideal special case, where the model is linear and Gaussian and the transport map is exact. This result is reinforced by numerical examples which show that the result seems to hold more generally. The significance of \ref{it:coupledResampling} is that to obtain a \gls{mse} of $\mathcal{O}(\epsilon^2)$ the cost is $\mathcal{O}(\epsilon^{-2})$; this is better than the \gls{mlpf}. Point \ref{it:smoother} relates to the afore-mentioned path degeneracy effect, which can mean \glspl{pf} (and hence the \gls{mlpf}) are not so useful in the context of large lag smoothing.

The structure of the article is as follows: \Cref{sec:multiLevelSDE} introduces the model and transport methodology. \Cref{sec:MLMC} presents the multilevel approach and the MLPF as well as the mechanisms underlying the computation of transport maps for a given level of discretization. The efficiency of the proposed approach is shown numerically on increasingly challenging scenarios in \cref{sec:numericalStudy}.

%%%%%%%%
\section[Methodology for SDE smoothing]{Methodology for \gls{sde} smoothing}
\label{sec:multiLevelSDE}

In this section, the considered notations and assumptions for the smoothing of \glspl{sde} are presented, together with a brief overview of the transport methodology.

%%%%
\subsection[The SDE model]{The \gls{sde} model}

Throughout the article, all random variables will be assumed to be on the same complete probability space $(\Omega,\Sigma,\bbP)$ and will be denoted by upper-case letters, while their realisations will be in lower case. We consider a diffusion process $\bsX = \{X_t\}_{t\in[0,T]}$ on the space $\bbR^d$ of the form
\eqnl{eq:diffusion}{
\d X_t = a(X_t) \d t + b(X_t)\d W_t ,\qquad t \in [0,T],
}
where $T$ is the final time, $\{W_t\}_{t \in [0,T]}$ is the Brownian motion on $\bbR^d$, $a(\cdot)$ is in the set $\calC^2(\bbR^d,\bbR^d)$ of twice continuously differentiable mappings from $\bbR^d$ to itself and $b(\cdot)$ is in $\calC^2(\bbR^d, \bbM_d(\bbR))$ with $\bbM_d(\bbR)$ the space of square matrices of size $d$. The mapping $b$ is assumed to be such that $b(x)b(x)^{\tr}$ is positive definite for all $x \in \bbR^d$, with $\cdot^{\tr}$ denoting the transposition. Moreover, the drift and diffusion coefficients are assumed to be globally Lipschitz, i.e.\ there exists $c > 0$ such that
\eqns{
|a(x) - a(x')| + |b(x) - b(x')| \leq c|x-x'|
}
for all $x,x' \in \bbR^d$. The initial distribution of the process $\bsX$, i.e.\ the distribution of $X_0$, is denoted $p_0$ (and might be equal to $\delta_{x_0}$ for some initial condition $x_0 \in \bbR^d$). It is assumed that the $m$\textsuperscript{th}-order moment of $X_0$ defined as $\bbE(|X_0|^m)$ is finite for any $m \geq 1$. Probability density functions will be considered with respect to the Lebesgue measure on $\bbR^d$ and both probability measures and their corresponding density functions will be referred to by the same notation.

The distribution of $X_k$, $k \in \{1,\dots,T\}$, given a realisation $x_{k-1}$ of the state $X_{k-1}$ is denoted $Q(x_{k-1},\cdot)$. In addition to the fact that the expression of the Markov transition $Q$ is unavailable in general, it is not usually possible to devise an unbiased estimator for it or even to sample from it. In the case where $d=1$, one can obtain ``skeletons'' of exact paths using the algorithm of \cite{Beskos2005,Beskos2006}, however, the extension of this approach to \glspl{sde} of higher dimensions might not be possible \cite{AitSahalia2008}.

The diffusion process $\bsX$ is assumed to be observed in $\bbR^{d'}$, $d' \in \bbN$, at all the integer-valued times so that the final time $T$ is also assumed to be an integer. These assumptions are made for the sake of notational simplicity and can be easily removed. For all $k \in \{0,\dots,T\}$, the observation $Y_k$ is a random variable that is conditionally independent on the state $X_t$ at times $t \neq k$ given $X_k$. The observation process can be expressed in general as
\eqnl{eq:obsEquation}{
Y_k = g_k(X_k, V_k)
}
where $g_k$ is a deterministic observation function and where $\{V_k\}_{k=0}^T$ is a collection of independent random variables. It is assumed without any real loss of generality that both $g_k$ and the distribution of $V_k$ do not depend on the time index $k$, the corresponding likelihood for a realisation $y_k$ of $Y_k$ is denoted $\ell(X_k, y_k)$.

%%%%
\subsection[Smoothing for SDEs]{Smoothing for \glspl{sde}}

Throughout the article, joint states in $\bbR^{d(n+1)}$ for some $n \in \bbN_0$ will be denoted either by $x_{k:k+n} \defeq (x_k,x_{k+1},\dots,x_{k+n})$ with $k \in \bbN$ or by $x_S$, with $S = \{s_0,s_1,\dots,s_n\}$ a finite subset of $[0,T]$ such that $s_i < s_j$ for all $0 \leq i < j \leq n$, defined as $x_S \defeq (x_{s_1},x_{s_2}, \dots,x_{s_n})$. The smoothing distribution associated with the \gls{sde} \cref{eq:diffusion} is defined formally as the joint law of the diffusion process $\bsX$ at all the integer times given realisations $y_0,\dots,y_T$ of the observation process \cref{eq:obsEquation}, and can be expressed for any $x_{0:T} \in \bbR^{d(T+1)}$ as
\eqns{
\bsp(x_{0:T}) = \dfrac{\ell(x_0, y_0) p_0(x_0) \prod_{k=1}^T \big[Q(x_{k-1}, x_k) \ell(x_k, y_k) \big] }{ \int \ell(x'_0, y_0) p_0(x'_0) \prod_{k=1}^T \big[Q(x'_{k-1}, x'_k) \ell(x'_k, y_k) \big] \d x'_{0:T} }.
}
The dependence of the smoothing distribution on the realisations $y_0,\dots,y_T$ of the observation process is omitted for the sake of notational simplicity. This is justified by the fact that these observations will be fixed in the remainder of the article so that the smoothing distribution $\bsp$ and its approximations will always be conditioned on the same given observations. The expression of $\bsp$ is a direct consequence of Bayes' theorem applied to the prior $p_0(x_0) \prod_{k=1}^T Q(x_{k-1}, x_k)$ describing the law of the unobserved (hidden) diffusion process together with the joint likelihood $\prod_{k=0}^T \ell(x_k, y_k)$ whose expression results from the conditional independence of the observations.

Using the same principle of implicit conditioning as with the smoothing distribution, the filtering distribution $p_k$ at time $k$ is defined as the law of $X_k$ given the realisations $y_0,\dots,y_k$ and is expressed recursively as
\eqns{
p_k(x_k) = \dfrac{\ell(x_k, y_k) \int Q(x_{k-1},x_k) p_{k-1}(x_{k-1}) \d x_{k-1}}{\int \ell(x'_k, y_k) Q(x'_{k-1},x'_k) p_{k-1}(x'_{k-1}) \d x'_k \d x'_{k-1}}
}
for any $x_k \in \bbR^d$ and any $k \in \{1,\dots,T\}$. The marginal distribution of $X_k$ induced by the smoothing distribution $\bsp$ corresponds to the filtering distribution $p_k$ when $k = T$ only.

The objective in this article can now be formally expressed as follows: to compute the expectation $\bsp(\varphi) \defeq \int \varphi(x_{0:T}) \bsp(x_{0:T}) \d x_{0:T}$ of some bounded measurable function $\varphi$ on $\bbR^{d(T+1)}$. Although the above formulation casts the considered problem into the standard Bayesian inference framework, the Markov transition $Q$ is unavailable in general, so that expressing analytically the distributions $\bsp$ and $p_k$ is not usually possible. The first step toward our objective is then to apply a time-discretization to the \gls{sde} \cref{eq:diffusion}, which, for the sake of simplicity, is illustrated with Euler's method for some discretization level $l \in \bbN_0$:
\eqnl{eq:EulerGen}{
X_{t+h_l} = X_t + h_l a(X_t) + \sqrt{h_l} b(X_t) U_t,
}
for some time-step $h_l = 2^{-l}$ and for all $t \in \calT_l \setminus \{T\}$ where $\calT_l \defeq \{0,h_l,\dots,T\}$, with $\{U_t\}_{t \in \calT_l \setminus \{T\}}$ a collection of independent Gaussian random variables with density $\phi(\cdot\,; 0,\bsI_d)$ where $\bsI_d$ is the identity matrix of size $d$. The choice of time step $h_l = 2^{-l}$ is made for the sake of convenience and is not necessary. The only requirement for both the \gls{mlpf} and the multilevel transport is that the ratio $h_{l-1}/h_l$ has to be an integer. The number of time steps from a given observation time up to and including the next observation time, that is in the interval $(k,k+1]$ for some $k \in \{0,\dots,T-1\}$, is $M_l = 2^l$. The numeral scheme \cref{eq:EulerGen} yields a Markov transition $K^l$ between two successive discretization times defined as
\eqns{
K^l(x, \cdot) = \phi\big(\cdot\, ; x + h_l a(x), h_l b(x)b(x)^{\tr}\big)
}
for any $x \in \bbR^d$, which enables the approximation of $Q$ by another Markov kernel $Q^l$ defined as
\eqns{
Q^l(x,\cdot) = \underbrace{K^l \dots K^l}_{M_l\text{ times}}(x,\cdot),
}
where $KK'(x,\cdot) = \int K(x,x') K'(x',\cdot) \d x'$ for any transition kernels $K$, $K'$. The smoothing distribution $\bsp^l$ induced by \cref{eq:EulerGen}, which approximates $\bsp$, is expressed on $\bbR^{d(M_lT+1)}$ instead of $\bbR^{d(T+1)}$ and is characterised by
\eqns{
\bsp^l(x_{\calT_l}) \propto p_0(x_0) \prod_{t \in \calT_l \setminus \{T\}} K^l\big(x_t, x_{t+h_l}\big) \prod_{k=0}^T \ell(x_k, y_k)
}
for any $x_{\calT_l} \in \bbR^{d(M_l T+1)}$. Marginalising \gls{wrt} all $x_t$ such that $t \notin \bbN_0$ gives a distribution on $\bbR^{d(T+1)}$ which depends on the same time steps as $\bsp$. It is understood that the error in the approximation of $Q$ and $\bsp$ by $Q^l$ and $\bsp^l$ decreases when $l$ increases and tend to $0$ as $l$ tends to infinity. The measure $\bsp^l(\varphi)$ of the function $\varphi$ is understood as the measure of the canonical extension $\bar\varphi$ of $\varphi$ from $\bbR^{d(T+1)}$ to $\bbR^{d(M_lT+1)}$ defined as
\eqns{
\bar\varphi(x_t) =
\begin{dcases*}
\varphi(x_t) & if $t \in \bbN_0$ \\
1 & otherwise.
\end{dcases*}
}
The extension $\bar\varphi$ of the function $\varphi$ can indeed be seen as canonical since it holds that
\eqnsa{
\bsp^l(\bar\varphi) & \propto \int \bar\varphi(x_{\calT_l}) p_0(x_0) \prod_{t \in \calT_l \setminus \{T\}} K^l\big(x_t, x_{t+h_l}\big) \prod_{k=0}^T \ell(x_k, y_k) \d x_{\calT_l} \\
& = \int \varphi(x_{0:T}) \ell(x_0, y_0) p_0(x_0) \prod_{k=1}^T \big[ Q^l(x_{k-1}, x_k) \ell(x_k, y_k) \big] \d x_{0:T},
}
as expected. Henceforth, $\bsp^l(\varphi)$ will be used has a shorthand notation for $\bsp^l(\bar\varphi)$ when there is no ambiguity.

At this stage, standard Bayesian inference methods can be easily applied. For instance, if $a$ and $b$ are linear and constant functions respectively and if the observation equation \eqref{eq:obsEquation} takes the form
\eqns{
Y_k = g_k(X_k) + V_k
}
with $g_k$ a linear map and with $V_k$ normally distributed, then the Kalman methodology can be used to determine the filtering and smoothing distributions. When this is not the case, the \gls{pf} methodology can be used instead, the approach exposed in \cite{Doucet2011} being one of the most popular versions. The latter applies sampling and resampling mechanisms to determine the filtering distribution with an error that is uniform in time. It is however less efficient for smoothing problems \cite{Kantas2015}, mostly because of the path degeneracy induced by the use of repeated resampling procedures.

The proposed second step toward the efficient computation of $\bsp(\varphi)$ is to use a method that enables i.i.d.\ samples to be drawn directly from the smoothing distribution $\bsp^l$ and hence avoiding path degeneracy. This has been made possible by transport methods \cite{Villani2008, Spantini2017} which are presented in the next section.

%%%%
\subsection{Transport methodology}
\label{sec:transport}

The general principle of transport methods, when applied to the considered problem, is to compute a deterministic coupling between the \emph{base} probability distribution $\bseta^l$ of a convenient i.i.d.\ process on $\bbR^d$ and the \emph{target} distribution $\bsp^l$, that is to compute a mapping $\bsG^l$ from $\bbR^{d(M_l T+1)}$ to itself that pushes forward $\bseta^l$ to $\bsp^l$, i.e.\ such that
\eqns{
\bsp^l(\bsx^l) = \bsG^l_{\pf} \bseta^l (\bsx^l) \defeq \bseta^l\big((\bsG^l)^{-1}(\bsx^l)\big) \big|\det \nabla (\bsG^l)^{-1}(\bsx^l) \big|,
}
where $\nabla (\bsG^l)^{-1}(\bsx^l)$ is the gradient of the inverse transport map $(\bsG^l)^{-1}$ evaluated at~$\bsx^l \in \bbR^{d(M_l T+1)}$. In this setting, the distribution $\bseta^l$ is also assumed to be on $\bbR^{d(M_l T+1)}$. The method introduced in \cite{Spantini2017} makes use of the specific structure of $\bsp^l$, which is induced by the Markov property of the underlying diffusion process~$\bsX$, to divide the problem into a sequence of low-dimensional couplings. Each of these deterministic couplings, say $M^l_t$ for some $t \in \calT_l \setminus \{T\}$, is a mapping from $\bbR^d \times \bbR^d$ to itself which is assumed to take the form
\eqns{
M^l_t : (x_t,x_{t+h_l}) \mapsto \big(M^{l,1}_t(x_t,x_{t+h_l}), M^{l,2}_t(x_{t+h_l})\big)^{\tr},
}
for some $M^{l,1}_t : \bbR^d \times \bbR^d \to \bbR^d$ and $M^{l,2}_t : \bbR^d \to \bbR^d$. Under additional assumptions on $M^{l,1}_t$ and $M^{l,2}_t$ (see \cref{eq:assumptionM1M2} below), the mapping $M^l_t$ can be characterised by
\eqns{
(M^l_t)_{\pf} \bseta^l_{t,t+h_l} = \bspi_{t,t+h_l},
}
where the probability distribution $\bseta^l_{t,t+h_l}$ on $\bbR^d \times \bbR^d$ is the marginal of $\bseta^l$ at discretization steps $(t,t+h_l)$ and where $\bspi_{t,t+h_l}$ is related to the marginal law of $(X_t,X_{t+h_l})$ and is characterised when $t > 0$ by
\eqns{
\bspi_{t,t+h_l}(x_t,x_{t+h_l}) \propto 
\begin{cases*}
\eta^l_t(x_t) K^l\big(M^{l,2}_{t-h_l}(x_t),x_{t+h_l}\big) \ell(x_{t+h_l}, y_{t+h_l}) & if $t+h_l \in \bbN$ \\
\eta^l_t(x_t) K^l\big(M^{l,2}_{t-h_l}(x_t),x_{t+h_l}\big) & otherwise,
\end{cases*}
}
where $\eta^l_t$ is the marginal of $\bseta^l$ on $\bbR^d$ at discretization time $t$, and by
\eqns{
\bspi_{0,h_l}(x_0,x_{h_l}) \propto
\begin{cases*}
p_0(x_0) K^0(x_0,x_1) \ell(x_0, y_0) \ell(x_1, y_1) & if $l = 0$ \\
p_0(x_0) K^l(x_0,x_{h_l}) \ell(x_0, y_0) & otherwise.
\end{cases*}
}

\begin{remark}
The expression of $\bspi_{t,t+h_l}$ at level $0$ is the one corresponding to the standard state space model presented in \cite{Spantini2017}, that is
\eqnsa{
\bspi_{t,t+1}(x_t,x_{t+1}) & \propto \eta_t(x_t) K\big(M^2_{t-1}(x_t),x_{t+1}\big) \ell(x_{t+1}, y_{t+1}), \qquad t > 0 \\
\bspi_{0,1}(x_0,x_1) & \propto p_0(x_0) K(x_0,x_1) \ell(x_0, y_0) \ell(x_1, y_1),
}
where the superscripts $0$ indicating the level have been omitted.
\end{remark}

The distribution $\bseta^l$ is a design variable which is chosen to be the normal distribution $\calN(0,\bsI_{d(M_lT+1)})$ for the sake of convenience (so that $\bseta^l_{t,t+h_l} = \phi(\cdot\,; 0,\bsI_{2d})$ and $\eta^l \defeq \eta^l_t = \phi(\cdot\,; 0,\bsI_d)$  do not depend on $t$). The two components of the mapping $M^l_t$ are instrumental for the proposed approach since they allow to transport samples from a convenient distribution to samples from the filtering or smoothing distributions. The filtering case is straightforward since it holds \cite[Theorem~7.1]{Spantini2017} that $M^{l,2}_t$ pushes forward $\eta^l_{t+h_l}$ to the filtering distribution $p^l_{t+h_l}$. To obtain samples from the smoothing distribution, it is necessary to first embed $M^l_t$ into the identity function on $\bbR^{d(M_lT+1)}$, which results in a function $G^l_t$ defined as
\eqns{
G^l_t : (x_0,x_{h_l},\dots,x_T) \mapsto \big(x_0,\dots,x_{t-h_l}, M^{l,1}_t(x_t,x_{t+h_l}) , M^{l,2}_t(x_{t+h_l}) , x_{t+2h_l},\dots,x_T \big)^{\tr}.
}
It is also demonstrated in \cite[Theorem~7.1]{Spantini2017} that the desired mapping $\bsG^l$, that is the one that pushes forward $\bseta^l$ to the smoothing distribution $\bsp^l$, is defined by the composition
\eqnl{eq:compositionOfMaps}{
\bsG^l = G^l_0 \circ G^l_{h_l} \circ \dots \circ G^l_{T-h_l}.
}

\begin{remark}
It would be possible to deduce a collection $\{\tilde{G}^{l-1}_t\}_t$ of transport maps at level $l-1$ by approximating pairwise compositions of maps at level $l$ as
\eqns{
\tilde{G}^{l-1}_t \approx G^l_t \circ G^l_{t+h_l}
}
for any $t \in \calT_{l-1} \setminus \{T\}$. However, it is less clear in this case which distribution is approximated by this new collection of transport maps.
\end{remark}

Although the transport maps $M^l_t$ have been identified, their computation is not straightforward. Assuming that the mappings $M^{l,1}_t$ and $M^{l,2}_t$ are of the form
\eqnl{eq:assumptionM1M2}{
M^{l,1}_t(x_{1:d},x'_{1:d}) =
\begin{bmatrix}
M^{l,1,1}_t(x_{1:d},x'_{1:d}) \\
\vdots \\
M^{l,1,d}_t(x_d,x'_{1:d})
\end{bmatrix}
\AND
M^{l,2}_t(x_{1:d}) =
\begin{bmatrix}
M^{l,2,1}_t(x_{1:d}) \\
\vdots \\
M^{l,2,d}_t(x_d)
\end{bmatrix},
}
for any $x_{1:d},x'_{1:d} \in \bbR^d$, i.e.\ loosely speaking, that $M^{l,1}_t$ and $M^{l,2}_t$ are upper triangular, it follows that $M^l_t$ is a $\sigma$-generalised Knothe-Rosenblatt (KR) rearrangement with $\sigma = (2d,2d-1,\dots,1)$, that is, informally, a map whose $i$\textsuperscript{th} component depends only on the variables $x_{2d},\dots,x_i$ and which pushes forward the $i$\textsuperscript{th} conditional of the base distribution to the corresponding conditional of the target distribution (see \cite[Definition~A.3]{Spantini2017} for more details). In order to find $M^l_t$, we first have to solve the following optimisation problem:
\eqnl{eq:opt_prob}{
M^{l,*} = \argmin_{M} - \bbE \bigg( \log \bspi_{t,t+h_l}(S_{\sigma}(M(\bsZ))) + \sum_{i=1}^{2d} \log \partial_i M^i(\bsZ) - \log \bseta^l_{t,t+h_l}(S_{\sigma}(\bsZ))  \bigg)
}
subject to $M$ being a monotone increasing lower triangular mapping, where the expectation is \gls{wrt} $\bsZ \sim \bseta^l_{t,t+h_l}$ and where $S_{\sigma}$ is the linear map corresponding to the transposition matrix induced by $\sigma$. It follows that $M^l_t = S_{\sigma} \circ M^{l,*} \circ S_{\sigma}$ since it holds that $S_{\sigma}^{-1} = S_{\sigma}$ for the considered permutation $\sigma$. The above optimisation problem can be solved in different ways, e.g.\ by Gauss quadrature or by having recourse to Monte Carlo techniques \cite{Robert2004, Davis2007}.

The transport map $\bsG^l$ enables an approximation of $\bsp^l(\varphi)$ to be computed by drawing $N$ samples $\{\bsz_i\}_{i=1}^N$ from $\bseta^l$ and by computing the empirical average
\eqns{
\tilde\bsp^l(\varphi) \defeq \dfrac{1}{N} \sum_{i=1}^N \varphi\big( \bsG^l(\bsz_i) \big) \approx \bsp^l(\varphi).
}
The \gls{mse} corresponding to the approximation of $\bsp(\varphi)$ by $\tilde\bsp^l(\varphi)$ can be expressed as the sum of a variance term and a bias term as follows
\eqns{
\bbE\big( (\tilde\bsp^l - \bsp)(\varphi)^2 \big) = \bbE\big( (\tilde\bsp^l - \bsp^l)(\varphi)^2 \big) + (\bsp^l - \bsp)(\varphi)^2.
}

We propose to further enhance the estimation by having recourse to a multilevel strategy for which transport methods will appear to be particularly well suited.

Although the method presented in this section applies in principle to state spaces of any dimension, it is important to note that the computational cost of the corresponding algorithm can be prohibitively high even for moderate dimensions. This issue can however be mitigated by identifying some specific dependence structure between the different dimensions and by applying the same principles as the ones applied here between time steps.

%%%%%%%%
\section{Multilevel Monte Carlo}
\label{sec:MLMC}

We now consider that the discretization \cref{eq:EulerGen} of the \gls{sde} \cref{eq:diffusion} is performed at different discretization levels $l \in \{0,\dots,L\}$ so that $0 < h_L < \dots < h_0 = 1$ for the considered value of $h_l$. This implies that the solution at the coarsest level $l=0$ is computationally efficient but possibly inaccurate whereas the solution at the finest level $L$ is more accurate but slower to compute. The principle of \gls{mlmc} is that the respective advantages of the coarsest and finest levels can be combined within a single estimation procedure by coupling the estimation of $\bsp(\varphi)$ for adjacent levels. More specifically, the first step is to notice that the smoothing distribution $\bsp^L$ corresponding to the discretization at level $L$ can be expressed via a telescopic sum involving the smoothing distributions $\bsp^l$ at the other levels $l < L$, that is
\eqnl{eq:telescopicSum}{
\bsp^L(\varphi) = \sum_{l=0}^L ( \bsp^l - \bsp^{l-1})(\varphi)
}
where $\bsp^{-1}$ is an arbitrary measure satisfying $\bsp^{-1}(\varphi) = 0$, e.g.\ the null measure. \Cref{eq:telescopicSum} motivates the introduction of some i.i.d.\ random variables $\{\bsX^0_i\}_{i=1}^{N_0}$ in $\bbR^{d(T+1)}$ with law $\bsp^0$ and some i.i.d.\ random variables $\{\bsX^{l,l-1}_i\}_{i=1}^{N_l}$ in the space $\bbR^{d(M_lT+1)} \times \bbR^{d(M_{l-1}T+1)}$ expressed as $\bsX^{l,l-1}_i = (\bsX^l_i,\bsX^{l-}_i)$ and such that $\bsX^l_i$ and $\bsX^{l-}_i$ have marginal laws $\bsp^l$ and $\bsp^{l-1}$ respectively, for all $l \in \{1,\dots,L\}$. This enables an approximation of $\bsp^L(\varphi)$ as
\eqnl{eq:telescopicSumApprox}{
\bsp^L(\varphi) \approx \tilde\bsp^L(\varphi) \defeq \dfrac{1}{N_0} \sum_{i=1}^{N_0} \varphi(\bsX^0_i) + \sum_{l=1}^L \dfrac{1}{N_l} \sum_{i=1}^{N_l} \big( \varphi(\bsX^l_i) - \varphi(\bsX^{l-}_i) \big).
}
This approximation of $\bsp^L$ is useful if the random variables $\bsX^0_{i_0}, \bsX^{1,0}_{i_1}, \dots, \bsX^{L,L-1}_{i_L}$ are independent of each other for all $i_0, i_1, \dots, i_L$ and if their respective components $\bsX^l_1$ and $\bsX^{l-}_1$ are as correlated as possible for all $l \in \{1,\dots,L\}$ (and hence for all random variables $\bsX^l_i$ and $\bsX^{l-}_i$ with $i \in \{1,\dots,N_l\}$ since they are i.i.d.). 

In order to determine the number of samples $N_l$ required at each level, we first express the \gls{mse} related to \cref{eq:telescopicSumApprox} as the sum of a variance term and a bias term as
\eqnl{eq:mseMl}{
\bbE\big( (\tilde\bsp^L - \bsp)(\varphi)^2 \big) = \sum_{l=0}^L \calV_l + (\bsp^L - \bsp)(\varphi)^2
}
with
\eqns{
\calV_l =
\begin{dcases*}
\bbE\Bigg( \bigg[\dfrac{1}{N_0} \sum_{i=1}^{N_0} \varphi(\bsX^0_i)  - \bsp^0(\varphi) \bigg]^2 \Bigg)& if $l=0$ \\
\bbE\Bigg( \bigg[\dfrac{1}{N_l} \sum_{i=1}^{N_l} \big( \varphi(\bsX^l_i) - \varphi(\bsX^{l-}_i) \big)  - (\bsp^l - \bsp^{l-1})(\varphi) \bigg]^2 \Bigg)& otherwise.
\end{dcases*}
}
Assuming that the bias is of order $\calO(h_L^{\alpha})$ for some integer $\alpha > 0$, it follows that a bias proportional to $\epsilon$ requires
\eqns{
L \propto -\dfrac{1}{\alpha}\log_2(\epsilon).
}
We also assume that the variance $\calV_l$ at level $l > 0$ is of order $\calO(h_l^{\beta})$ and that the cost $\calC_l$ at level $l$ is of order $\calO(h_l^{-\zeta})$ for some positive integers $\beta$ and $\zeta$. The number of samples $N_{l}$ at level $l > 1$ can then be determined by optimising the total cost $\calC = \sum_l \calC_l N_l$ for a given total variance $\calV = \sum_l \calV_l / N_l$. This leads to
\eqnl{eq:Nl}{
N_l = N_1 2^{-(\beta + \zeta)(l-1)/2},
}
so that, to obtain a \gls{mse} of order $\epsilon^2$, that is a bias of order $\epsilon$ and a total variance of order $\epsilon^2$, one must take $N_0 \propto \epsilon^{-2}$ and
\eqns{
N_1 \propto \epsilon^{-2} \sum_{l = 1}^L 2^{(\zeta - \beta)l/2}.
}
Therefore, the number of samples and the cost for a \gls{mse} of order $\calO(\epsilon^2)$ depends on the respective values of $\beta$ and $\zeta$. For instance, if $\beta > \zeta$, then both $N_1$ and $\calC$ are of order $\calO(\epsilon^{-2})$.

%%%%
\subsection{Multilevel particle filter}

It is assumed in this section that the interest lies in estimating the filtering distribution $p^L_k$ at time $k$ through the multilevel identity~\cref{eq:telescopicSum}. Since it is generally difficult to sample directly from a reasonable candidate for a coupling of $p^l_k$ and $p^{l-1}_k$, one solution is to adopt a \gls{pf} strategy within the \gls{ml} formulation. In order to obtain samples that are correlated between two adjacent levels, a special joint Markov transition $Q^{l,l-1}$ can be devised together with a resampling procedure that retains the correlation of the samples. This is the principle of the \gls{mlpf} which is briefly discussed here. Assume that we have some collections of samples $\{x^l_{i,k-1}\}_{i=1}^{N_l}$ and $\{x^{l-}_{i,k-1}\}_{i=1}^{N_l}$ at time $k-1$ approximating $p^l_{k-1}$ and $p^{l-1}_{k-1}$ respectively. For all $i \in \{1,\dots,N_l\}$ and all $l \in \{1,\dots,L\}$, samples $x^l_{i,k}$ and $x^{l-}_{i,k}$ at time $k$ are produced through the Markov transition $Q^{l,l-1}((x^l_{i,k-1},x^{l-}_{i,k-1}),\cdot)$ as follows:
\begin{enumerate}[label=(\roman*)]
\item Simulate \cref{eq:EulerGen} starting from the initial condition $x_0 = x^l_{i,k-1}$ over $M_l$ time steps, denote by $x^l_{i,k}$ the obtained state of the process and by $\{u^l_t\}_{t\in\{0,h_l,\dots,1-h_l\}}$ the collection of realisations of the perturbation $U^l_t$ drawn during the procedure.
\item Using the initial condition $x^{l-}_0 = x^{l-}_{i,k-1}$, define $x^{l-}_{i,k}$ as the result of the deterministic recursion
\eqns{
x^{l-}_{t+h_{l-1}} = x^{l-}_t + h_{l-1} a(x^{l-}_t) + \sqrt{h_{l-1}} b(x^{l-}_t) ( u^l_{t} + u^l_{t+h_l} ),
}
for any $t \in \{0,h_{l-1},\dots,1-h_{l-1}\}$. This recursion is meaningful since $h_{l-1} = 2h_l$ so that $u^l_{t} + u^l_{t+h_l}$ corresponds to the noise in the step from $t$ to $t+h_{l-1}$ induced by $\{u^l_t\}_t$.
\end{enumerate}
This procedure yields $N_l$ pairs of correlated samples $\{(x^l_{i,k}, x^{l-}_{i,k})\}_{i=1}^{N_l}$ according to the predictive distribution at time $k$ given observations up to time $k-1$. The information provided by the observation $y_k$ is simply taken into account by attributing the respective weights $w^l_{i,k}$ and $w^{l-}_{i,k}$ to the samples $x^l_{i,k}$ and $x^{l-}_{i,k}$ in a similar fashion:
\eqns{
w^l_{i,k} = \dfrac{\ell(x^l_{i,k}, y_k)}{ \sum_{j=1}^{N_l} \ell(x^l_{j,k}, y_k) } \AND
w^{l-}_{i,k} = \dfrac{\ell(x^{l-}_{i,k}, y_k)}{ \sum_{j=1}^{N_l} \ell(x^{l-}_{j,k}, y_k) }.
}
Following the weighting of the samples, the difference $(p^l_k - p^{l-1}_k)(\varphi)$ can be estimated via
\eqns{
(p^l_k - p^{l-1}_k)(\varphi) \approx \sum_{i=1}^{N_l} \Big( w^l_{i,k}\varphi\big(x^l_{i,k}\big) - w^{l-}_{i,k}\varphi\big(x^{l-}_{i,k}\big) \Big).
}
Although this approximation would behave well in general, most of the sample weights would tend to $0$ if we were to apply the same procedure repeatedly in order to reach the next observation times, resulting in a rapid increase of the empirical variance. The usual way to address this problem in the standard \gls{pf} formulation is to perform resampling, that is to draw new samples from the old ones according, for instance, to the multinomial distribution induced by the weights. Applying the same approach to the \gls{mlpf} would result in the loss of the correlation between the samples at adjacent levels. A \emph{coupled} resampling is used instead as follows. For all $i \in \{1,\dots,N_l\}$ and all $l \in \{1,\dots,L\}$:
\begin{enumerate}[label=(\roman*)]
\item \label{it:coupledIndex} With probability $\rho^l_k = \sum_{i=1}^{N_l} \min\{w^l_{i,k}, w^{l-}_{i,k}\}$ draw the index $i^l$ according to the probability mass function (p.m.f.) $\hat{m}^l_k$ on $\{1,\dots,N_l\}$ characterised by
\eqns{
\hat{m}^l_k(j) = \dfrac{1}{\rho^l_k} \min\{w^l_{j,k}, w^{l-}_{j,k}\}
}
and define $i^{l-} = i^l$.
\item If \ref{it:coupledIndex} is not selected (with probability $1-\rho^l_k$), draw the indices $i^l$ and $i^{l-}$ independently according to the p.m.f.s $m^l_k$ and $m^{l-}_k$ on $\{1,\dots,N_l\}$ characterised by 
\eqns{
m^l_k(j) \propto w^l_{j,k} - \min\{w^l_{j,k}, w^{l-}_{j,k}\} \AND m^{l-}_k(j) \propto w^{l-}_{j,k} - \min\{w^l_{j,k}, w^{l-}_{j,k}\}.
}
\item Define the new pair of samples $(\tilde{x}^l_{i,k},\tilde{x}^{l-}_{i,k})$ as $(x^l_{i^l,k},x^{l-}_{i^{l-},k})$.
\end{enumerate}
Although the \emph{coupled} resampling addresses the problem of reducing the empirical variance without completely losing the correlation between samples at adjacent levels, it nevertheless has a negative impact of the \gls{ml} rate. Indeed, as demonstrated in \cite{Jasra2015}, one needs $\beta > 2\zeta$ to obtain a cost of order $\calO(\epsilon^{-2})$ for a \gls{mse} of order $\calO(\epsilon^2)$. In the case where $\beta = 2\zeta$, e.g.\ for Euler's scheme ($\zeta = 1$) with $\beta = 2$, the cost is of order $\calO(\epsilon^{-2} \log(\epsilon)^2)$.

Also, even if the \gls{mlpf} can handle smoothing on a short time window, i.e.\ it can successfully approximate the distribution of $\{X_{t'}\}_{t' \in \{t-s,t-s+1,\dots,t\}}$ given $y_0,\dots,y_t$ for small values of $s \in \bbN$, the error in the approximation of the full smoothing distribution would increase in time because of the path degeneracy effect. Indeed, resampling tends to multiply the samples of higher weights so that, after a certain number of time steps, all samples will be descendants of the same earlier sample.

%%%%
\subsection{Multilevel transport}

In order to avoid the path degeneracy inherent to any \gls{pf} approach and to regain the \gls{ml} rate lost through the coupled resampling of the \gls{mlpf}, we propose to compute samples from the distributions $\bsp^l$ via the transport maps $\bsG^l$ characterised by $\bsp^l = \bsG^l_{\pf} \bseta^l$ with $\bseta^l = \phi(\cdot\,; 0,\bsI_{d(M_lT+1)})$ for all $l \in \{0,\dots,L\}$. The specific procedure is described as follows. For all $i \in \{1,\dots,N_l\}$:
\begin{enumerate}[label=(\roman*)]
\item draw a sample $\bsz_i^l = (z_{i,0}^l,z_{i,1}^l,\dots,z_{i,M_lT}^l)$ from $\bseta^l$
\item map $\bsz_i^l$ through $\bsG^l$ to obtain a sample $\bsx_i^l = \bsG^l(\bsz_i^l)$ from $\bsp^l$
\item define a \emph{thinned} sample $\bsz_i^{l-} = (z_{i,0}^l,z_{i,2}^l,\dots,z_{i,M_lT}^l)$
\item map $\bsz_i^{l-}$ through $\bsG^{l-1}$ to obtain a sample $\bsx_i^{l-} = \bsG^{l-1}(\bsz_i^{l-})$ from $\bsp^{l-1}$
\end{enumerate}
This simple procedure yields two collections $\{\bsx_i^l\}_i$ and $\{\bsx_i^{l-}\}_i$ of samples drawn from a joint distribution that obviously has marginals $\bsp^l$ and $\bsp^{l-1}$ and that correlates adjacent levels as desired. As a motivation for this coupling, note that it is optimal in terms of squared Wasserstein distance with the Euclidean metric in the case where $d=1$ and assuming that the transport maps can be computed exactly. The efficiency of the approach comes from the fact that the transport maps $\bsG^l$ have to be computed once only. Given the computation of the maps, it is relatively fast to obtain the samples. %Once the maps computed, the obtention of samples is comparatively fast.

Although there is, strictly speaking, no path degeneracy in the considered approach, there might be some accumulation of error through time induced by the composition of transport maps defining $\bsG^l$ as in \cref{eq:compositionOfMaps}. This accumulation of error will however be seen to be milder than the one experienced by the \gls{pf} in \cref{sec:numericalStudy}.

It is assumed that the procedure underlying the computation of the transport maps is deterministic, so that there is no undesired correlations between samples from $\bsX^{l,l-1}$ and $\bsX^{l',l'-1}$ when $l \neq l'$. Further neglecting the numerical error in the computed transport maps, it follows that the expression \cref{eq:mseMl} of the \gls{mse} holds for the considered approach.

Before proceeding to a numerical study, the legitimacy of the proposed approach is verified for the linear-Gaussian case. Consider the \gls{sde} \cref{eq:diffusion} in dimension $d=1$ and with $p_0 = \delta_{x_0}$ (so that the observation at time $t=0$ has no impact). The corresponding filtering distribution at time $k \in \bbN$ and at level $l \in \{0,\dots,L\}$ simplifies to
\eqns{
p^l_k(x_k) \propto \int \prod_{n=1}^k \big[ Q^l(x_{n-1},x_n) \ell(x_n, y_n) \big] \d x_{1:k-1}
}
for any $x_k \in \bbR^d$. Denote $\hat{G}^l_k \defeq M^{l,2}_{k-h_l}$ the transport map from the base distribution $\eta^l = \phi(\cdot\,; 0,1)$ to $p^l_k$, i.e.\ such that $(\hat{G}^l_k)_{\#} \eta^l = p^l_k$. If $F_{\eta^l}$ and $F_{l,k}$ denote the cumulative distribution functions (c.d.f.) of $\eta^l$ and $p^l_k$ respectively, then it holds that $\hat{G}^l_k = F^{-1}_{l,k} \circ F_{\eta^l}$, where $F^{-1}$ is the generalised inverse
\eqns{
F^{-1}(u) = \inf\{x \in \bbR : F(x) \geq u\}, \qquad \forall u \in [0,1].
}

Considering i.i.d.\ random variables $Z_i \sim \eta^l$ for $i \in \{1,\dots,N_l\}$, the objective is to determine the order of
\eqns{
\calV_{l,k} = \Var\bigg( \dfrac{1}{N_l} \sum_{i=1}^{N_l} \Big( \varphi\big(\hat{G}^l_k(Z_i) \big) - \varphi\big( \hat{G}^{l-1}_k(Z_i) \big) \Big) \bigg)
}
\gls{wrt} $h_l$ for any function $\varphi$ that is at the intersection of the set $\calB_b(\bbR)$ of bounded measurable functions and of the set $\Lip(\bbR)$ of Lipschitz functions. Since the $Z_i$'s are i.i.d.\ and by definition of $\hat{G}^l_k$, it holds that
\eqnsa{
\calV_{l,k} & = \dfrac{1}{N_l} \Var\Big( \varphi\big(\hat{G}^l_k(Z) \big) - \varphi\big( \hat{G}^{l-1}_k(Z) \big) \Big) \\
& = \dfrac{1}{N_l} \Var\Big( \varphi\big(F^{-1}_{l,k}(U) \big) - \varphi\big( F^{-1}_{l-1,k}(U) \big) \Big) \\
& \leq \dfrac{c}{N_l} \bbE\Big( \big[F^{-1}_{l,k}(U) - F^{-1}_{l-1,k}(U) \big]^2 \Big)
}
for some $c > 0$, with $Z \sim \eta^l$ and $U \sim \calU([0,1])$, where the inequality comes from the fact that $\varphi \in \Lip(\bbR)$. The linear case is addressed in the following theorem as a proof of concept.

\begin{theorem}
\label{res:orderVarLinearGaussian}
Let $\bsX$ a $1$-dimensional diffusion process with linear drift and constant diffusion coefficient observed at all integer times through a linear-Gaussian likelihood $\ell(x, \cdot) = \phi(\cdot\,; x,\tau^2)$ for some $\tau > 0$, then the variance $\calV_{l,k}$ obtained at level $l$ for Euler's method with discretization $h_l = 2^{-l}$ and with the transport-based approach satisfies
\eqns{
\calV_{l,k} = \calO(h_l^2)
}
for any $k \in \{1,\dots,T\}$.
\end{theorem}

\begin{proof}
The objective is to compute the order of
\eqns{
F^{-1}_{l,k}(u) - F^{-1}_{l-1,k}(u) = \hat\mu_{l,k} - \hat\mu_{l-1,k} + \sqrt{2}\erf^{-1}(2u-1) (\hat\sigma_{l,k} - \hat\sigma_{l-1,k})
}
\gls{wrt} $h_l$, where $\erf^{-1}$ is the inverse error function and where the updated mean $\hat\mu_{l,k}$ and standard deviation $\hat\sigma_{l,k}$ at level $l$ and at time $k$ can be found through the Kalman filter to be
\eqnsa{
\hat\mu_{l,k} = \mu_{l,k} + \dfrac{\sigma^2_{l,k}(y - \mu_{l,k})}{\tau^2 + \sigma_{l,k}^2} \AND
\hat\sigma_{l,k}^2 = \dfrac{\tau^2\sigma_{l,k}^2}{\tau^2 + \sigma_{l,k}^2}
}
with $\mu_{l,k}$ and $\sigma_{l,k}$ the predicted mean and standard deviation expressed as
\eqns{
\mu_{l,k} = (1 + h_l a)^{M_l}\hat\mu_{l,k-1} \AND
\sigma_{l,k}^2 = (1 + h_l a)^{2M_l}\hat\sigma^2_{l,k-1} + h_l b^2 \sum_{i=0}^{M_l-1}(1 + h_l a)^{2i}.
}
First, the predicted mean $\mu_{l,k}$ and standard deviation $\sigma_{l,k}$ have to be developed to the second order. The main term appearing in the expressions of $\mu_{l,k}$ is
\eqns{
(1 + h_l a)^{M_l} = \sum_{n=0}^{M_l} \dfrac{a^n}{n!} \prod_{i=0}^{n-1} \big[ h_l (M_l - i)\big] = \sum_{n=0}^{M_l} \dfrac{a^n}{n!} + \dfrac{h_l}{2} \sum_{n=2}^{M_l} \dfrac{a^n}{(n-2)!} + \calO(h_l^2),
}
For the sake of compactness we define
\eqns{
A_m = \sum_{n=0}^m \dfrac{a^n}{n!} \AND B_m = \sum_{n=2}^m \dfrac{a^n}{(n-2)!}.
}
Assuming that
\eqnmla{eq:proof:assumedForm}{
\hat{\mu}_{l,k-1} & = c_{k-1} + r_{k-1,l} h_l +  \calO(h_l^2) \\
\hat{\sigma}_{l,k-1} & = c'_{k-1} + r'_{k-1,l} h_l + \calO(h_l^2)
}
where $c_{k-1}$ and $c'_{k-1}$ do not depend on $l$, and where $r_{k-1,l}$ and $r'_{k-1,l}$ are of order $\calO(1)$ \gls{wrt} $h_l$, it follows that
\eqnsa{
\mu_{l,k} & = \hat\mu_{l,k-1} \Big( A_{M_l} + \dfrac{h_l}{2} B_{M_l} \Big) + \calO(h_l^2) \\
& = c_{k-1}  A_{M_l} + r_{k-1,l} h_l A_{M_l} + h_l\dfrac{c_{k-1}}{2} B_{M_l} + \calO(h_l^2).
}
Recalling that $M_l = 2^l$ and noticing that
\eqns{
A_{M_l} = e^a - \sum_{n \geq M_l + 1} \dfrac{a^n}{n!} = e^a + o(h_l)
}
with $o(h_l)$ referring to terms that are negligible in front of $h_l$, $\mu_{k,l}$ can be seen to be of the same form as $\hat\mu_{k,l}$, that is
\eqns{
\mu_{l,k} = c_{k-1} e^a + r_{k-1,l} h_l e^a + h_l\dfrac{c_{k-1}}{2} B_{M_l} + \calO(h_l^2).
}
The same type of expansion can be used for the first term in the variance $\sigma_{l,k}^2$ as follows 
\eqns{
\sigma_{l,k}^2 = c'^2_{k-1} e^a + 2c'_{k-1} r'_{k-1,l} h_l e^a + h_l\dfrac{c'^2_{k-1}}{2} B_{2^{l+1}} + b^2 h_l \sum_{i=0}^{M_l-1}(1 + h_l a)^{2i} + \calO(h_l^2).
}
The second term has however a slightly different form and must be studied on its own:
\eqnl{eq:secondTermStdDev}{
h_l\sum_{i=0}^{M_l-1}(1 + h_l a)^{2i} = \sum_{n = 0}^{2M_l - 2} \Bigg( h_l^{n+1} a^n \sum_{i=\lceil n/2 \rceil}^{M_l - 1} \binom{2i}{n} \Bigg)
}
where it appears that
\eqns{
h_l^{n+1} a^n\sum_{i=\lceil n/2 \rceil}^{M_l-1} \binom{2i}{n} \leq h_l^{n+1} a^n \sum_{i=1}^{M_l} \dfrac{(2i)^n}{n!} = \dfrac{(2a)^n}{(n+1)!}
}
where the r.h.s.\ tends exponentially fast to $0$ when $n \to \infty$. It follows that \cref{eq:secondTermStdDev} is of the form $s + o(h_l)$ where $s$ does not depend on $l$, so that
\eqns{
\sigma_{l,k}^2 = c'^2_{k-1} e^a + 2c'_{k-1} r'_{k-1,l} h_l e^a + h_l\dfrac{c'^2_{k-1}}{2} B_{2^{l+1}} + s b^2 + \calO(h_l^2),
}
from which the expansion of the standard deviation $\sigma_{l,k}$ can be expressed as
\eqns{
\sigma_{l,k} = \sqrt{C_l} + \dfrac{h_l}{2\sqrt{C_l}}\bigg( 2c'_{k-1} r'_{k-1,l} h_l e^a + \dfrac{c'^2_{k-1}}{2} B_{2^{l+1}} \bigg) + \calO(h_l^2) 
}
where $C_l = e^a c'^2_{k-1} + s b^2$ is the term of order $\calO(1)$ in $\sigma_{l,k}^2$. We conclude that
\eqns{
\mu_{l,k} - \mu_{l-1,k} = h_l \big(r_{k-1,l} A_{2^l}  - 2 r_{k-1,l-1} A_{2^{l-1}} \big) + h_l \dfrac{c_{k-1}}{2} \big(B_{2^l}  - 2B_{2^{l-1}} \big) + \calO(h_l^2) = \calO(h_l).
}
Similarly, it holds that $\sigma_{l,k} - \sigma_{l-1,k} = \calO(h_l)$. Proceeding to the updated terms, it holds that
\eqnsa{
\sigma^2_{l,k}(y_k - \mu_{l,k}) & = (e^a c'^2_{k-1} + s b^2) (y_k - c_{k-1} e^a) + \calO(h_l) \\
\tau^2 + \sigma^2_{l,k} & = \tau^2 + (e^a c'^2_{k-1} + s b^2) + \calO(h_l),
}
so that
\eqnsa{
\hat\mu_{l,k} & = \mu_{l,k} + \dfrac{\sigma^2_{l,k}(y_k - \mu_{l,k})}{\tau^2 + \sigma^2_{l,k}}
= c_{k-1} e^a + \dfrac{(e^a c'^2_{k-1} + s b^2) (y_k - c_{k-1} e^a)}{\tau^2 + (e^a c'^2_{k-1} + s b^2)} + \calO(h_l)\\
\hat\sigma_{l,k} & = \dfrac{\tau^2\sigma^2_{l,k}}{\tau^2 + \sigma^2_{l,k}} = \dfrac{\tau^2 (e^a c'^2_{k-1} + s b^2)}{\tau^2 + (e^a c'^2_{k-1} + s b^2)} + \calO(h_l).
}
If follows from reasoning by induction that $\hat\mu_{l,k}$ and $\hat\sigma_{l,k}$ have the form assumed in \cref{eq:proof:assumedForm} for all $k \in \{0,\dots,T\}$, the result being obvious for $k=0$. Combining the different results it can be easily verified that
\eqnsa{
\hat\mu_{l,k} - \hat\mu_{l-1,k} = \calO(h_l) \AND \hat\sigma_{l,k} - \hat\sigma_{l-1,k} = \calO(h_l),
}
which yields $\calV_{l,k} = \calO(h_l^2)$ as desired. This concludes the proof of the \lcnamecref{res:orderVarLinearGaussian}.
\end{proof}

%%%%%%%%
\section{Numerical study}
\label{sec:numericalStudy}

In this section, the effectiveness of the proposed method is shown in simulations for different \gls{sde} models. Numerical verifications of some of the considered assumptions are also provided. The scenarios considered for simulation are the same as for the \gls{mlpf} in \cite{Jasra2015}, so that results can be compared.

%%%%%
\subsection{Linear Gaussian}

The first simulation study is performed on the linear-Gaussian case with $a=-0.1$, $b=1$ and with a likelihood $\ell(x, \cdot) = \phi(\cdot\,; x,\tau^2)$ with $\tau = 0.25$ which corresponds to an observation process of the form
\eqnl{eq:linearObservation}{
Y_k \given X_k \sim \calN(0,\tau^2).
}
The initial distribution is $p_0 = \phi(\cdot\,; 0,\sigma)$ with $\sigma = 1$ and the final time is $T=4$. A realisation of the state and observation processes are shown in \cref{fig:fourLevels} together with the mean and some percentiles corresponding to samples drawn from the smoothing distribution. The involved transport maps\footnote{The solver used for the determination of the transport maps is the one provided at \url{http://transportmaps.mit.edu/docs/index.html}}, say $T$, are assumed to be triangular maps which $i$\textsuperscript{th} component $T^{(i)}$ takes the form
\eqns{
T^{(i)}(x_1,\dots,x_i) = a_i(x_1,\dots,x_{i-1}) + \int_0^{x_i} b_i(x_1, \dots, x_{i-1},t)^2 \d t
}
where $a_i$ and $b_i$ are real-valued functions defined on $\bbR^{i-1}$ and $\bbR^i$ respectively. For any $j \leq i-1$, it is assumed that the functions $x_j \mapsto a_i(x_1,\dots,x_{i-1})$ and $x_j \mapsto b_i(x_1,\dots,x_{i-1},t)$ are Hermite Probabilists' functions extended with constant and linear components whereas the function $t \mapsto b_i(x_1,\dots,x_{i-1},t)$ is assumed to be a Hermite Probabilists' function extended with a constant component only. Then, the functions $a_i$ and $b_i$, when expressed as functions from $\bbR^{i-1}$ and $\bbR^i$ respectively, take the form
\begin{align*}
a_i(x_1,\dots,x_{i-1}) & = \sum_{k = 1}^{2d(o_{\mathrm{m}}+1)} c_k \Phi_k(x_1,\dots,x_{i-1}) \\
b_i(x_1,\dots,x_{i-1},t) & = \sum_{k = 1}^{2do_{\mathrm{m}}} c'_k \Psi_k(x_1,\dots,x_{i-1},t)
\end{align*}
with $o_{\mathrm{m}}$ the map order, with $\{c_k\}_{k \geq 1}$ and $\{c'_k\}_{k \geq 1}$ some collections of real coefficients and with $\Phi_k$ and $\Psi_k$ basis functions based on the above mentioned Hermite Probabilists' functions. In the simulations, the case $o_{\mathrm{m}} = 4$ is considered.

The integration in \cref{eq:opt_prob} is performed using a Gauss quadrature of order $10$ in each dimension. The optimisation relies on the Newton-CG algorithm (Newton algorithm using the conjugate-gradient method for each step) with a tolerance of $10^{-4}$.

\begin{figure}
\centering
\includegraphics[trim=70pt 70pt 100pt 90pt,clip,width=.8\textwidth]{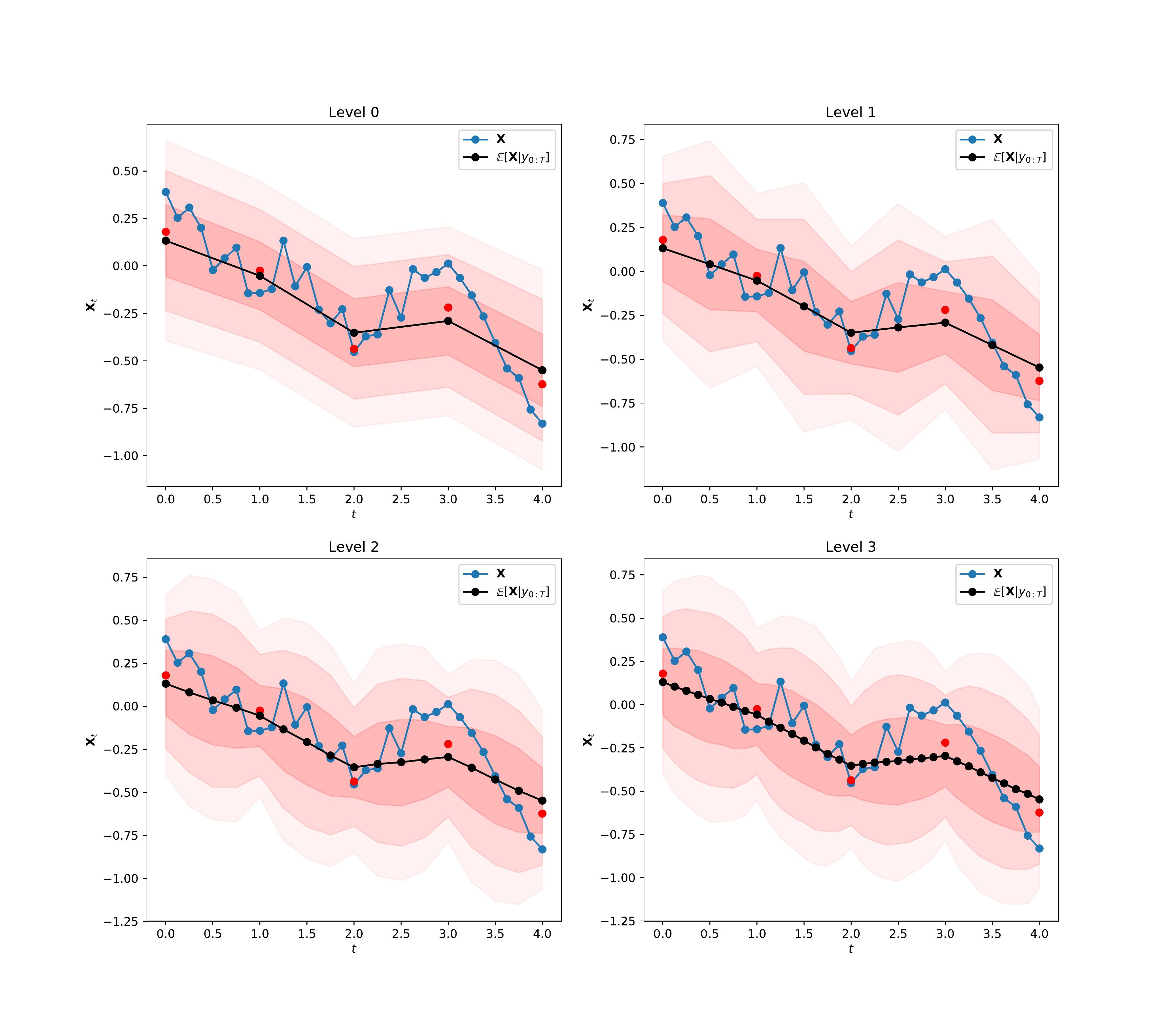}
\caption{Mean and percentiles of samples generated according to the target distribution of the linear-Gaussian \gls{sde} at four consecutive levels (blue line: state of the process; red dots: observations; black line: samples mean; red areas: $1$-$99$, $5$-$95$ and $20$-$80$ percentiles).}
\label{fig:fourLevels}
\end{figure}

\subsubsection*{\gls{mlmc} rates} The behaviour of the numerical scheme for different levels is displayed in \cref{fig:varDiff_cost}, where $\Var(\varphi(\bsX^l) - \varphi(\bsX^{l-1}))$ is considered with $\varphi(x_{0:T}) = x_T$ and where the cost is the computational time required to obtain one sample at a given level $l$. This result confirms the applicability of multilevel techniques by showing that $\calV_l = \calO(h_l^2)$ and $\calC_l = \calO(h_l^{-1})$, that is $\beta = 2$ and $\zeta = 1$.

One important point is that the time spent to obtain samples at a high level is small when compared to the time required to compute the underlying transport map. For instance, it takes about $25\s$ to calculate the transport map at level $5$ while a sample is obtained in $0.00025\s$, so that a $100,000$ samples can be drawn in the time spent to compute the map. It is therefore necessary to verify that the gain obtained with the multilevel approach is not compensated by the additional time spent computing more transport maps (one for each level).

\begin{figure}
\centering
\subfloat[Linear-Gaussian]{%
    \label{fig:varDiff_cost}%
    \includegraphics[width=0.495\textwidth]{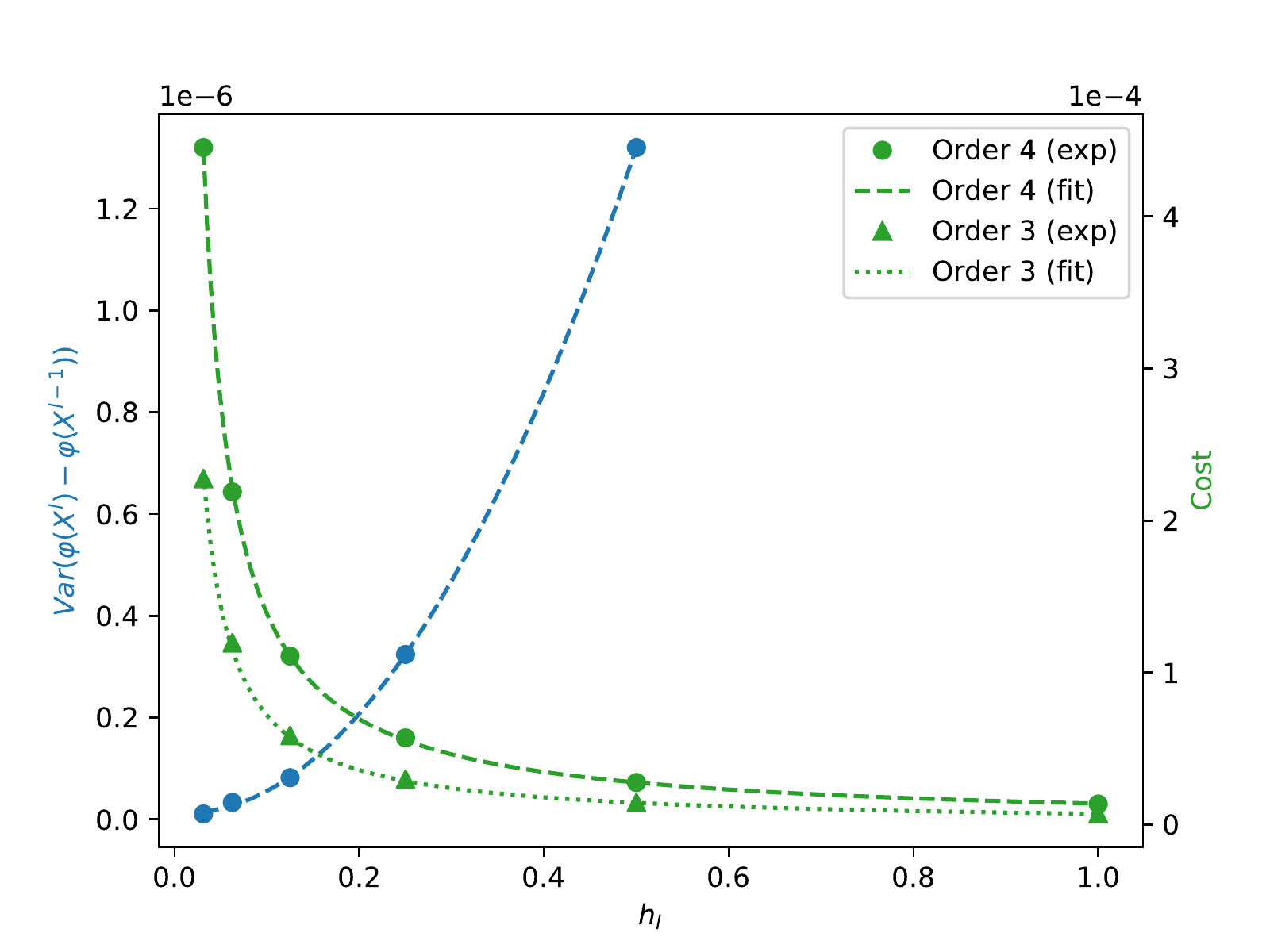}%
}%
\subfloat[Langevin]{%
    \label{fig:varDiff_cost_LD}%
    \includegraphics[width=0.495\textwidth]{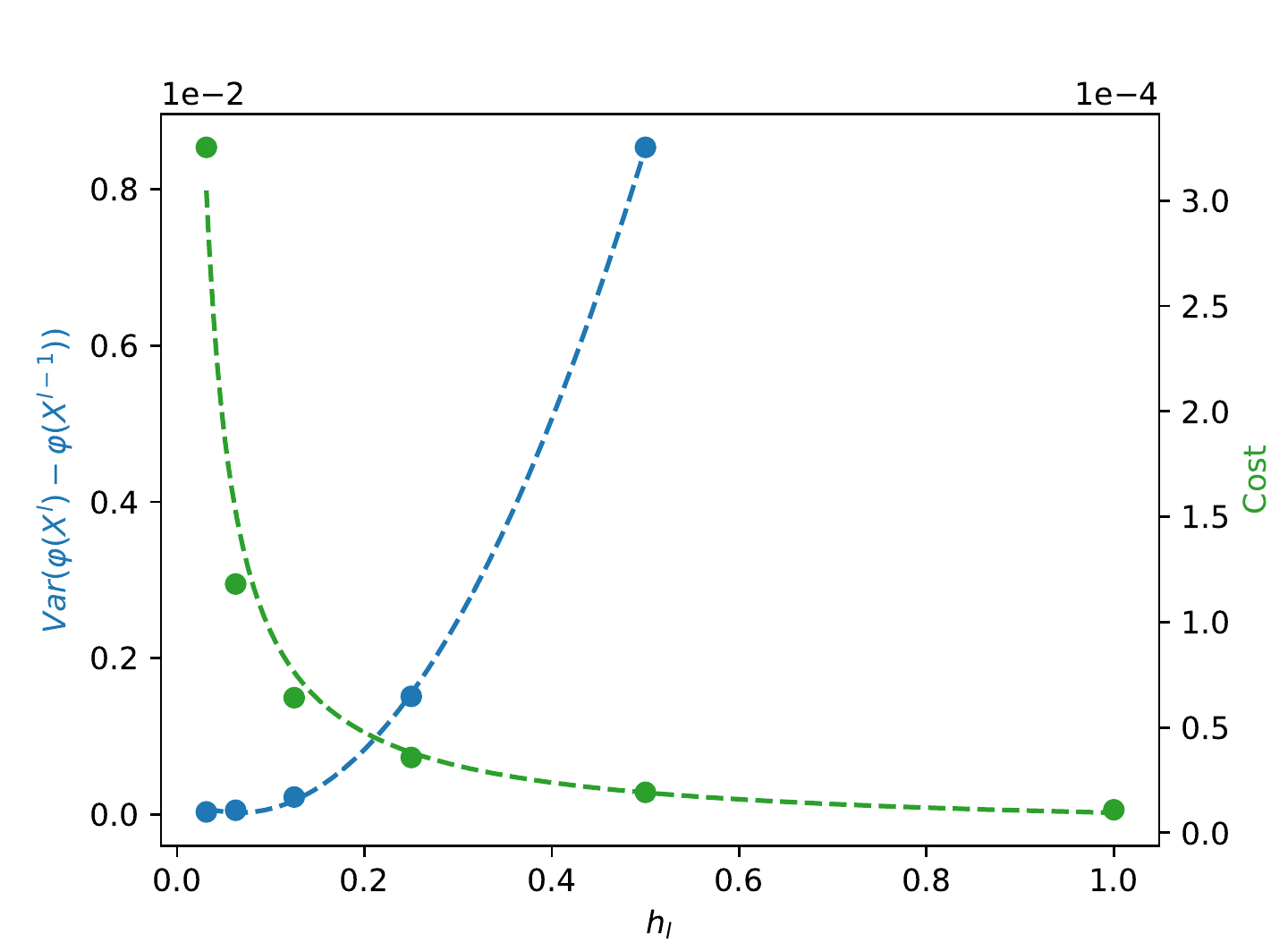}%
}\\
\subfloat[Non-linear diffusion]{%
    \label{fig:varDiff_cost_NLD}%
    \includegraphics[width=0.495\textwidth]{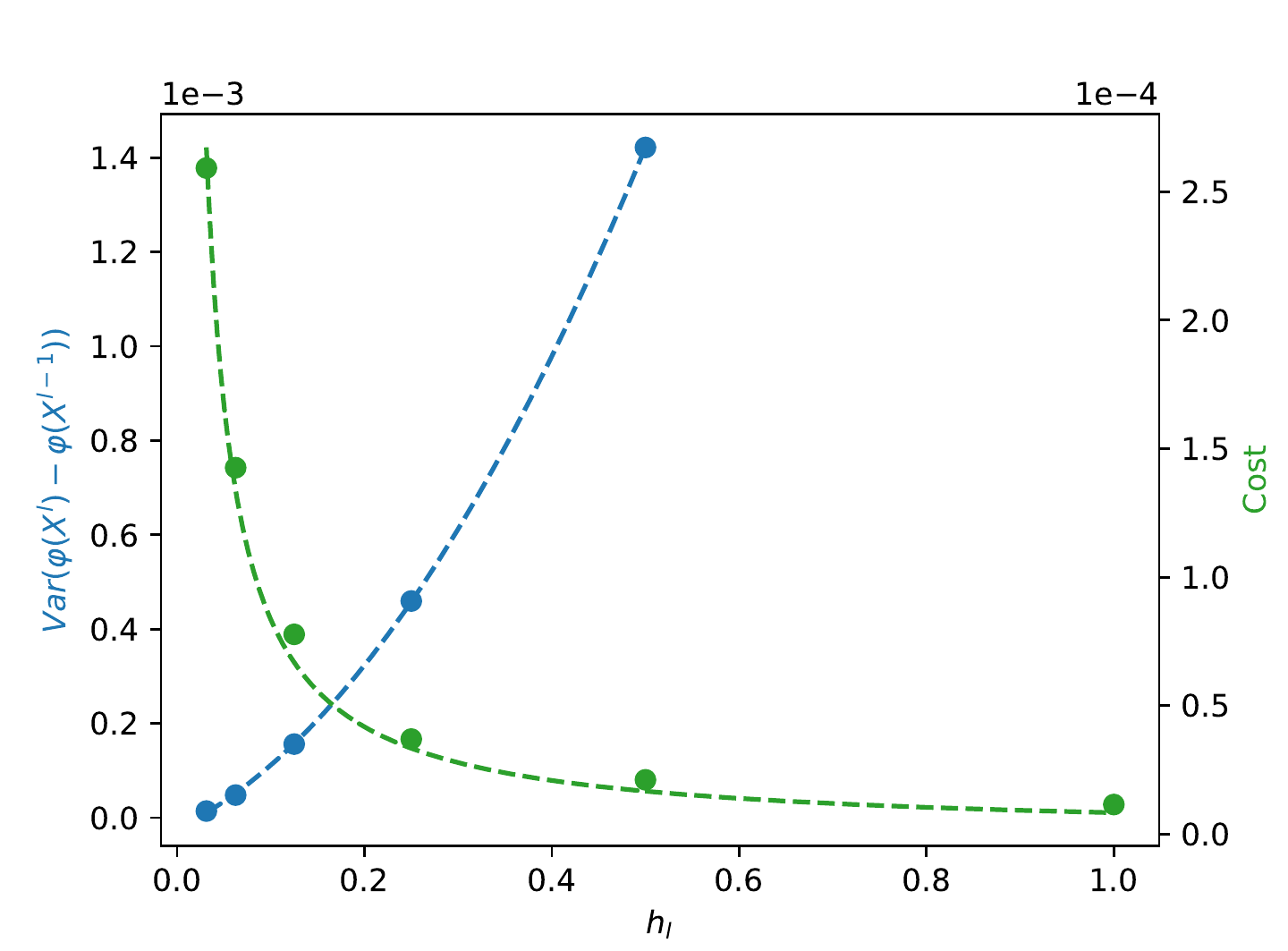}%
}%
\caption{Variance of $\varphi(\bsX^l) - \varphi(\bsX^{l-1})$ with $\varphi(x_{0:T}) = x_T$ and cost as a function of $h_l$  (Blue dashed line: poly.\ fit of order $2$; green dashed line: least-square fitting of the form $a/h_l$). The experimental (exp) cost for two map-approximation orders are indicated in the linear-Gaussian case together with their corresponding least-square fittings (fit).}
\label{fig:varDiff_cost_all}
\end{figure}

\subsubsection*{Multilevel vs computation at the highest level} The objective with the multilevel approach is to reduce the computational cost to reach a given error when compared to computations at the highest level only. This aspect is verified in \cref{fig:MLvsHL_LG} where the multilevel approach appears to outperform the one based on samples at the highest level. The above-mentioned fact that calculation of the transport maps might be time-consuming is shown to be compensated by the efficiency of the multilevel approach within a reasonable time interval. This is in spite of the fact that the multilevel approach nearly doubles the number of maps to be computed. In particular, in the considered linear-Gaussian scenario, the average computational cost for the calculation of the maps in the multilevel and highest-level approach is respectively $10.76\s$ and $6.15\s$.

More specifically, \cref{fig:MLvsHL} is obtained by first computing all the required transport maps and then by generating samples by batches of $1000$. The multilevel estimate is obtained by sweeping the different levels sequentially until the predetermined number $N_l$ of samples has been computed at level $l$. The number $N_0$ of samples at level $0$ is fixed to $2^{13} \times 1000$ for all the considered \glspl{sde}, that is $2^{13}$ batches of $1000$ samples. The number of samples at level $1$ is determined by the ratio between the variance at levels $0$ and $1$ and the number of samples for the subsequent levels are computed through \cref{eq:Nl}.

\begin{figure}
\centering
\subfloat[Linear-Gaussian $\varphi(x_{0:T}) = x_T$]{%
    \label{fig:MLvsHL_LG}%
    \includegraphics[width=0.495\textwidth]{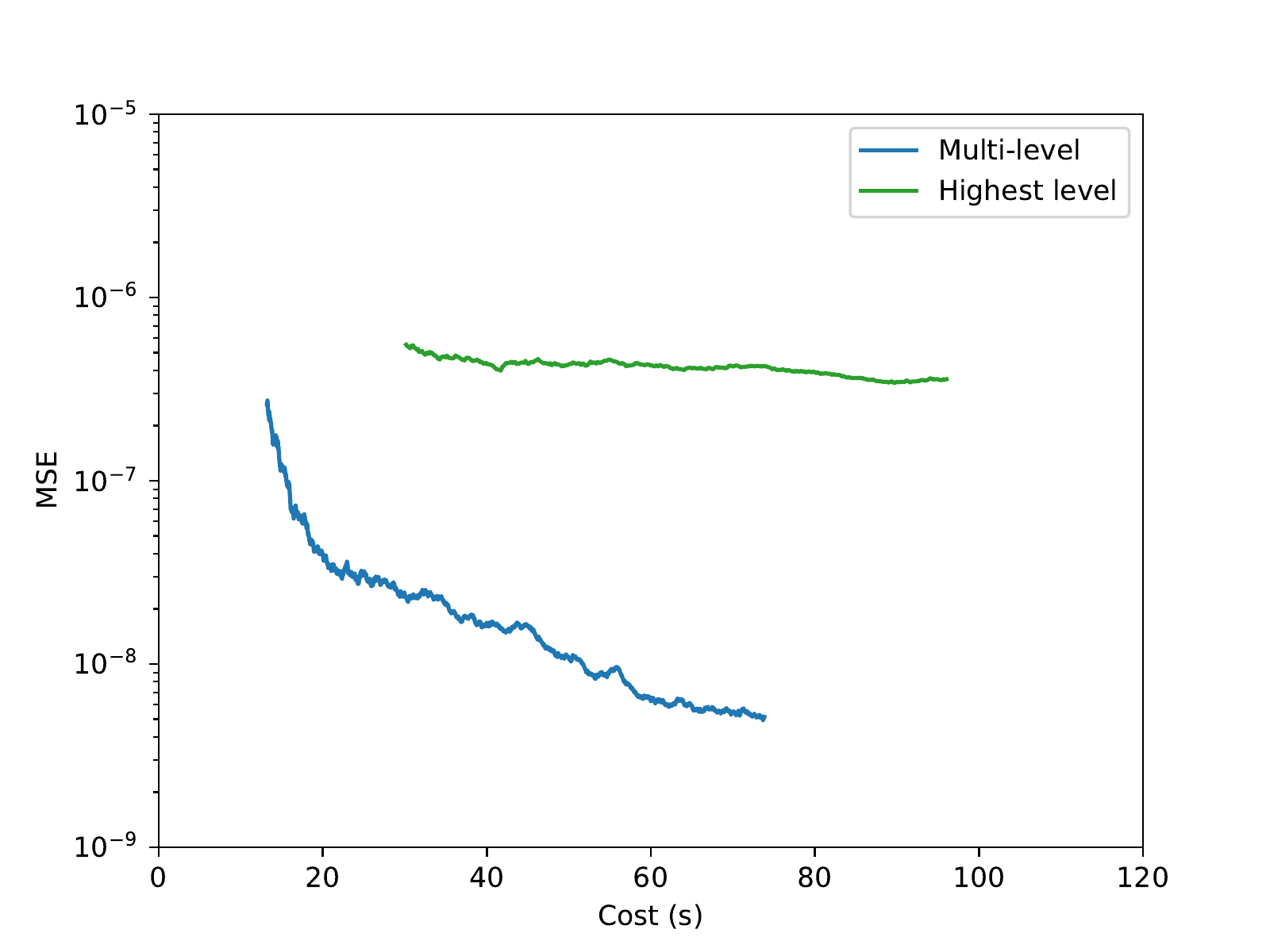}%
}%
\subfloat[Langevin $\varphi(x_{0:T}) = x_T$]{%
    \label{fig:MLvsHL_LD}%
    \includegraphics[width=0.495\textwidth]{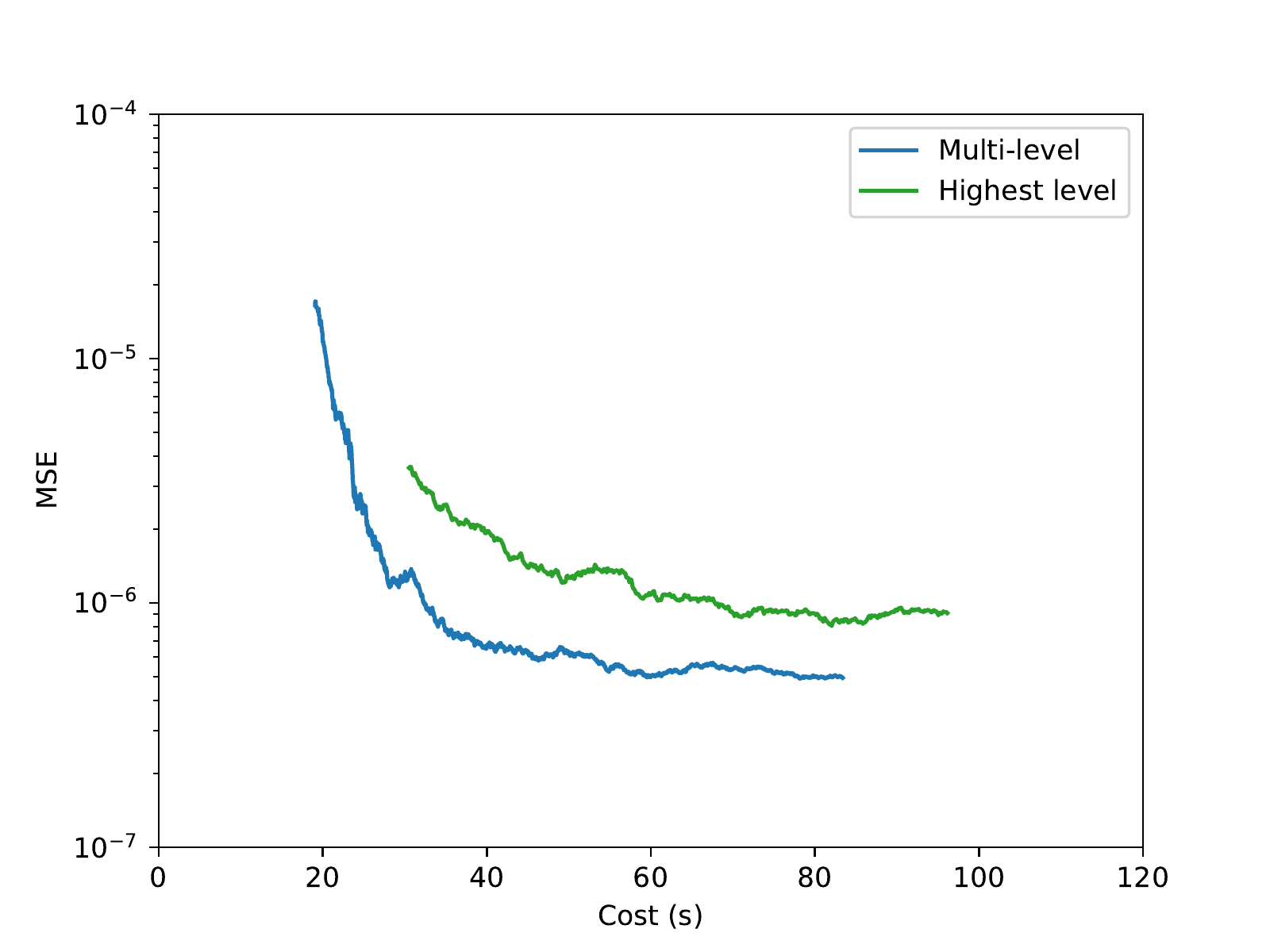}%
}\\
\subfloat[Langevin $\varphi(x_{0:T}) = \sum_{t=0}^T e^{-\kappa(T-t)}x_t$]{%
    \label{fig:MLvsHL_LD_2}%
    \includegraphics[width=0.495\textwidth]{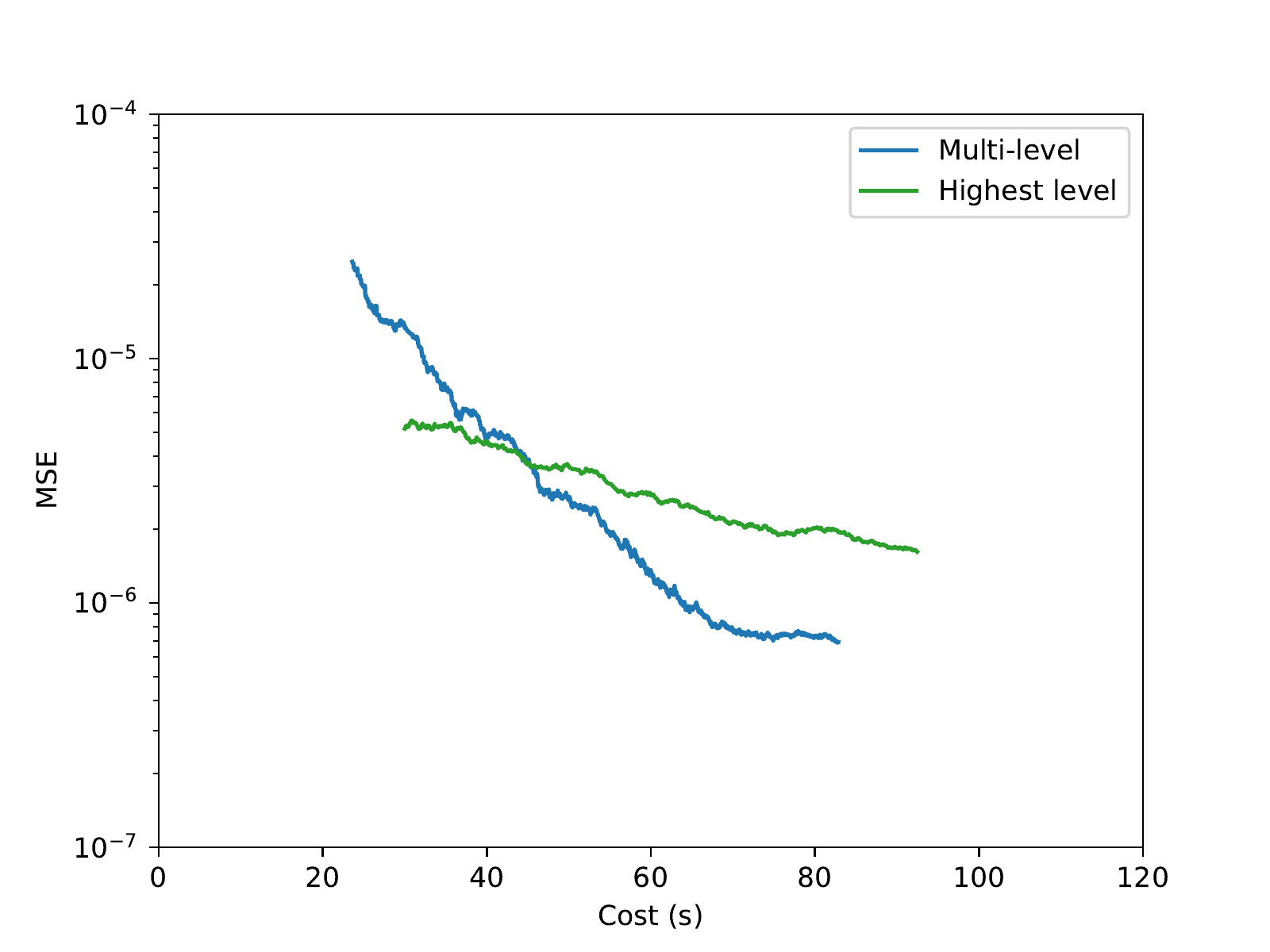}%
}%
\subfloat[Non-linear diffusion $\varphi(x_{0:T}) = x_T$]{%
    \label{fig:MLvsHL_NLD}%
    \includegraphics[width=0.495\textwidth]{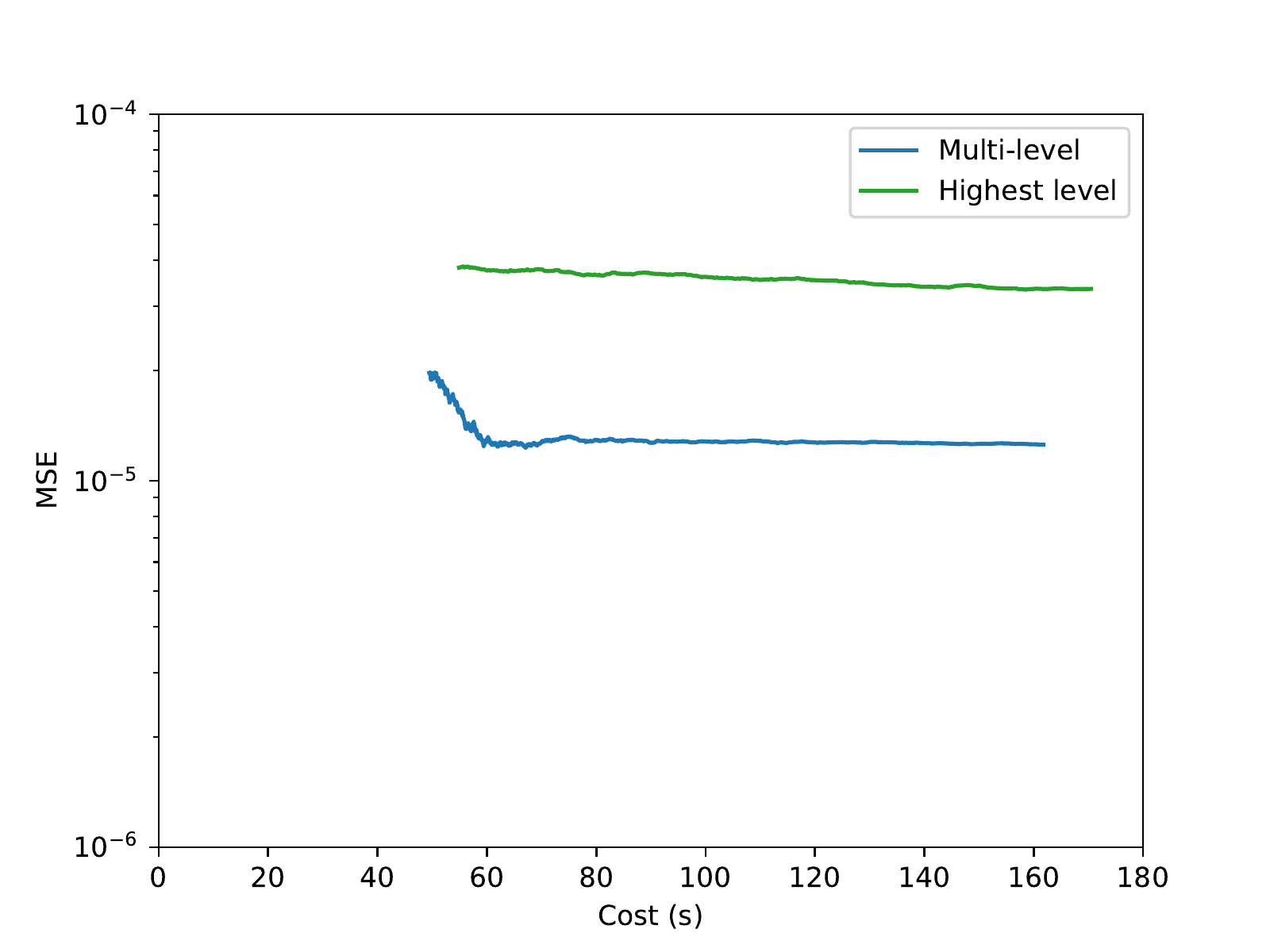}%
}%
\caption{MSE vs.\ cost for the multilevel approach compared with computations at the highest level $L=4$ (semi-log scale, averaged over 50 Monte Carlo simulations). The first $200$ iterations are not displayed.}
\label{fig:MLvsHL}
\end{figure}

%%%%%
\subsection[Langevin SDE]{Langevin \gls{sde}}

We now consider a Langevin \gls{sde} of the form
\eqns{
\d X_t = \dfrac{1}{2} \nabla \log \calS_{\nu} (X_t) \d t + b\, \d W_t, \qquad t \in [0,T]
}
where $\calS_{\nu}$ is the Student's t distribution with $\nu=10$ degrees of freedom and with $b = 1$. The observations are generated according to
\eqnl{eq:obsExpVar}{
Y_k \given X_k \sim \calN\big(0,\tau^2 \exp(X_k)\big)
}
with $\tau = 1$. The initial distribution is the same as in the previous example. A realisation of the considered Langevin \gls{sde} is shown in \cref{fig:fourLevels_LD} together with mean and percentiles of samples obtained using transport maps. It appears clearly on this figure that the observation process characterised by \cref{eq:obsExpVar} is less informative than the one modelled by \cref{eq:linearObservation}. \Cref{fig:varDiff_cost_LD} shows that the considered Langevin \gls{sde} also displays a variance of order $\calO(h_l^2)$, although the actual values are much higher than in the linear-Gaussian case, which might be due to both the nature of the \gls{sde} and the quality of the approximation of the transport maps. A comparison of the computational efficiency of the multilevel approach is given in \cref{fig:MLvsHL_LD} where the proposed method is seen to outperform the approach based on computations at the highest level. The time needed to initialise the latter, i.e.\ the time to compute the transport map at level $L=4$ and to perform the first $200$ iterations, is however slightly less affected than with the multilevel approach. \Cref{fig:MLvsHL_LD_2} shows the performance of the proposed approach with a different functional, that is
\eqns{
\varphi(x_{0:T}) = \sum_{t=0}^T \exp(-\kappa(T-t))x_t,
}
which gives the sum of the states at the observations weighted by a forgetting factor~$\kappa$, with $\kappa = 2$ in the simulations. In this case, the tolerance of the optimisation is also adapted to the level as follows: the tolerance at level $l$ is $10^{-l-1}$. This helps retaining the benefits of the multi-level approach in this more challenging smoothing problem.

\begin{figure}
\centering
\includegraphics[trim=70pt 70pt 100pt 90pt,clip,width=.8\textwidth]{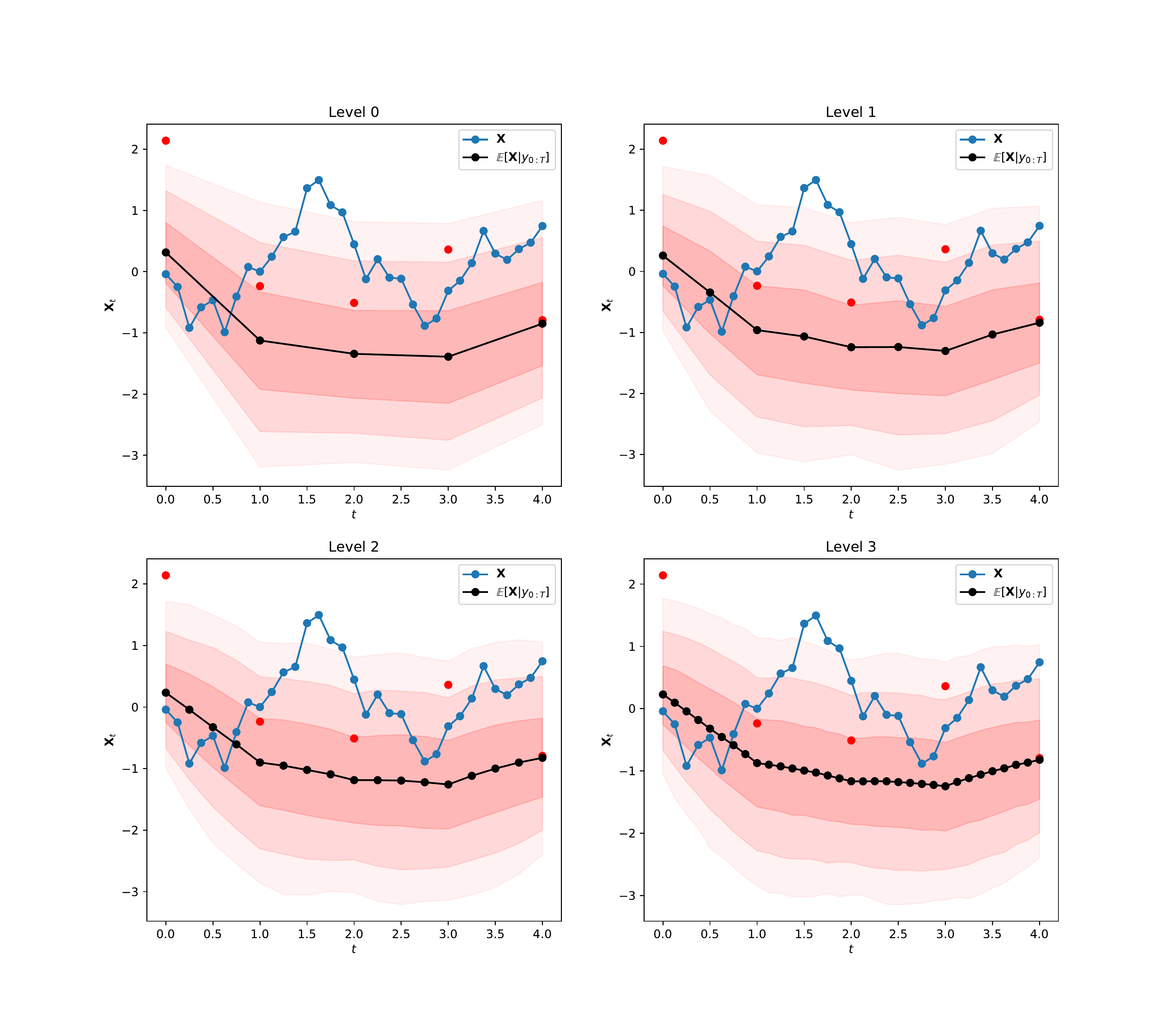}
\caption{Mean and percentiles of samples generated according to the target distribution of the Langevin \gls{sde} at four consecutive levels (blue line: state of the process; red dots: observations; black line: samples mean; red areas: $1$-$99$, $5$-$95$ and $20$-$80$ percentiles).}
\label{fig:fourLevels_LD}
\end{figure}

%%%%%
\subsection{Nonlinear diffusion}

We now consider a \gls{sde} with a nonlinear diffusion term:
\eqns{
\d X_t = \theta (\mu - X_t) \d t + \dfrac{\varsigma}{\sqrt{1+X_t^2}} \d W_t, \qquad t \in [0,T]
}
with $\theta = 1$, $\mu = 1$ and $\varsigma = 1$ and with a time step of $0.5$ between observation times, so that the final time is $T=2$. The linear-Gaussian observation model \cref{eq:linearObservation} is considered with $\tau = 1$. The initial distribution is the same as in the previous examples. A realisation of the considered \gls{sde} is displayed in \cref{fig:fourLevels_NLD} together with mean and percentiles of samples obtained using transport maps. \Cref{fig:varDiff_cost_NLD} shows that the same rates as in the previous cases apply although the contribution of the quadratic term in the variance is smaller than before. It appears in \cref{fig:MLvsHL_NLD} that the time spent computing the transport maps has largely increased for both approaches when compared to the linear-Gaussian and Langevin \glspl{sde}. This might be due to the challenging nature of the problem which induces a slower convergence on the involved optimisation methods. However, the proposed method still displays a significant gain in performance, although the first 200 iterations just gave it enough time to compensate for the computational overhead caused by the calculation of the maps at all level.

\begin{figure}
\centering
\includegraphics[trim=70pt 70pt 100pt 90pt,clip,width=.8\textwidth]{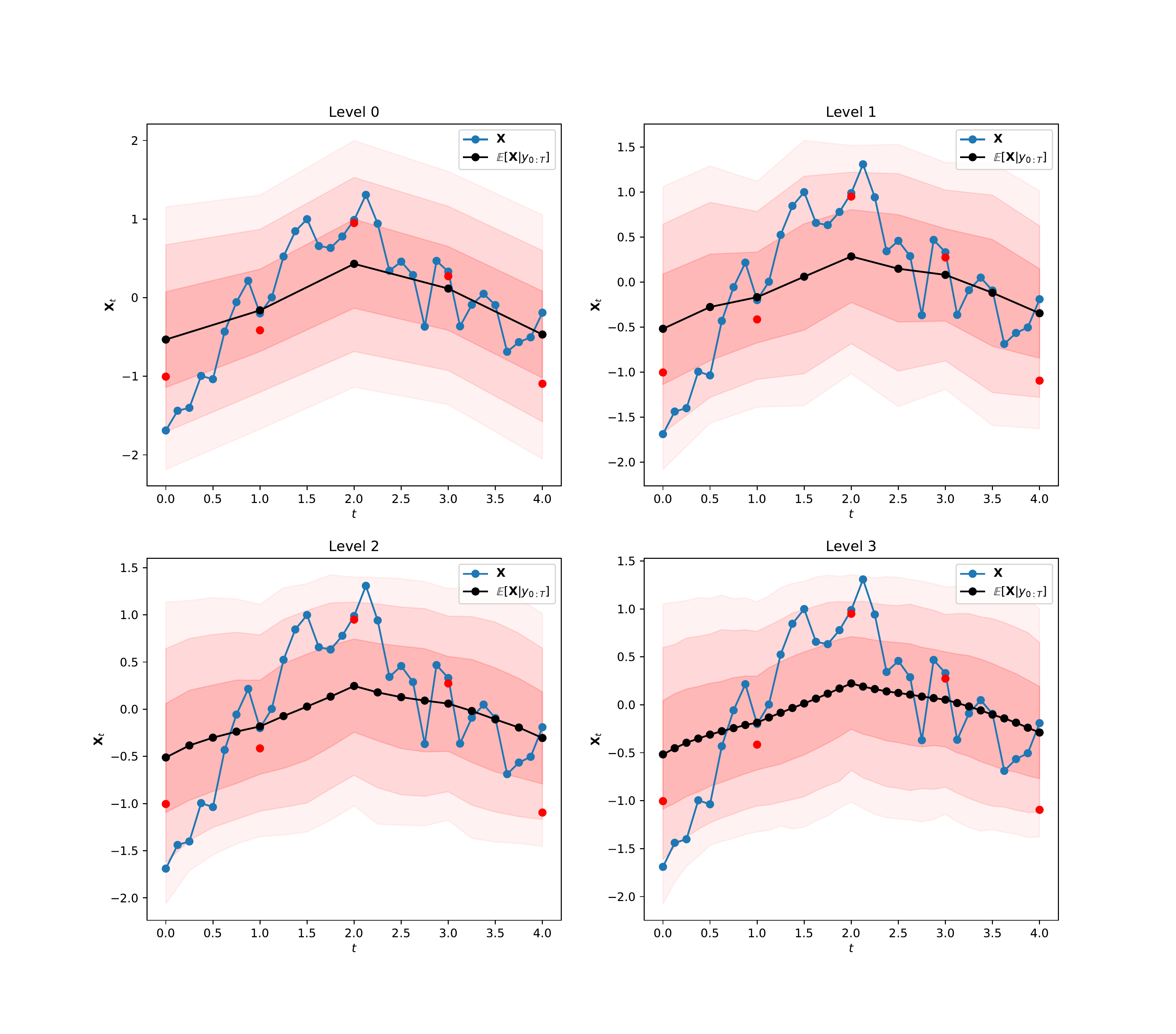}
\caption{Mean and percentiles of samples generated according to the target distribution of the \gls{sde} with nonlinear diffusion at four consecutive levels (blue line: state of the process; red dots: observations; black line: samples mean; red areas: $1$-$99$, $5$-$95$ and $20$-$80$ percentiles).}
\label{fig:fourLevels_NLD}
\end{figure}

%%%%%%%%
\section{Conclusion}

An algorithm for the determination of expectations with respect to laws of partially-observed \glspl{sde} has been proposed. The observations are received at discrete times and depend only on the state at the time they occurred, hence enabling a standard state space modelling to be used. The proposed method relies on three principles:
\begin{enumerate*}[label=(\roman*)]
\item the discretization of the considered \gls{sde}, for instance with Euler's method,
\item the expression of the smoothing distribution at a given level as a telescopic sum involving coarser discretizations and
\item the generation of pairs of samples correlated across adjacent levels via the application of different transport maps to samples from a common base distribution.
\end{enumerate*}
As opposed to \gls{mlpf}, the proposed approach retains the ``ideal'' \gls{mlmc} rates, since, in particular, it does not require resampling techniques to be used. In addition to a numerical verification of its performance, the proposed method was shown to have the desired behaviour in the linear-Gaussian case. Future works include the theoretical verification of the rates that are observed in practice for more diverse types of \glspl{sde}, as well as the study of the optimal parametrisation of the transport maps as a function of the discretization level.

\subsection*{Acknowledgements}

The authors would like to thank the Associate Editor as well as the referees for their detailed comments and suggestions for the manuscript. All authors were supported by Singapore Ministry of Education AcRF tier 1 grant R-155-000-182-114. AJ was also supported under KAUST CRG4 Award Ref:2584. AJ is affiliated with the Risk Management Institute, OR and analytics cluster and the Center for Quantitative Finance at NUS.  

\bibliographystyle{siamplain}
\bibliography{Transport}

\end{document}